\documentclass[12pt]{article}
\usepackage{amsfonts}
\usepackage{amsmath}
\usepackage{amsthm}
\usepackage{amssymb}
\usepackage{graphicx}
\usepackage{subfigure}
\usepackage{indentfirst}

\newtheorem{theorem}{Theorem}
\newtheorem{lemma}{Lemma}
\newtheorem{proposition}[theorem]{Proposition}

\newcommand{\mev}{\mbox{MeV}}
\newcommand{\kev}{\mbox{keV}}

\newcommand{\psip}{\psi^{\prime}}

\newcommand{\jpsi}{J/\psi}

\newcommand{\TT}{\tau^+\tau^-}

\newcommand{\thetast}{\theta^{\ast}}
\newcommand{\phist}{\phi^{\ast}}
\newcommand{\jst}{J_{\ast}}
\newcommand{\rst}{r^{\ast}}
\newcommand{\sigmast}{\sigma^{\ast}}
\newcommand{\taums}{m_{\tau}}
\newcommand{\sigbg}{\sigma_{BG}}

\newcommand{\beq}{\begin{equation}}
\newcommand{\eeq}{\end{equation}}

\def\eref#1{(\ref{#1})}
\def\lmref#1{{Lemma \ref{#1}}}
\def\thref#1{{Theorem \ref{#1}}}
\newsavebox{\arrect}
\savebox{\arrect}(0,0){%
\setlength{\unitlength}{0.5cm} \thicklines
\put(-4,1){\line(1,0){8.0}} \put(4,1){\line(0,-1){2.0}}
\put(4,-1){\line(-1,0){8.0}} \put(-4,-1){\line(0,1){2.0}}
\put(0,-1){\vector(0,-1){1.0}} }

\newsavebox{\arrhomb}
\savebox{\arrhomb}(0,0){%
\setlength{\unitlength}{1cm} \thicklines
\put(0,0.5){\line(-2,-1){1.0}} \put(0,0.5){\line(2,-1){1.0}}
\put(0,-0.5){\line(-2,1){1.0}} \put(0,-0.5){\line(2,1){1.0}}
\put(0,-0.5){\vector(0,-1){0.5}}
\put(-0.15,-0.7){\makebox(0,0)[r]{no}}
\put(1.0,0.0){\vector(1,0){2.0}}
\put(2.0,0.05){\makebox(0,0)[b]{yes}} }
\newsavebox{\arrparall}
\savebox{\arrparall}(0,0){%
\setlength{\unitlength}{0.5cm} \thicklines
\put(-3,1){\line(1,0){8.0}} \put(5,1){\line(-1,-1){2.0}}
\put(3,-1){\line(-1,0){8.0}} \put(-5,-1){\line(1,1){2.0}}
\put(0,-1){\vector(0,-1){1.0}} }
\newsavebox{\arrparalla}
\savebox{\arrparalla}(0,0){%
\setlength{\unitlength}{0.5cm} \thicklines
\put(-3,1){\line(1,0){8.0}} \put(5,1){\line(-1,-1){2.0}}
\put(3,-1){\line(-1,0){8.0}} \put(-5,-1){\line(1,1){2.0}} }

\begin{document}

\title{Theory of second optimization for scan experiment }

\author{X.H.Mo\footnote{E-mail:moxh@ihep.ac.cn}\\
{\small Institute of High Energy Physics, Chinese Academy of
    Sciences, Beijing 100049, China } }

\date{\today}
\maketitle

\begin{abstract}
In many high energy experiments, the physics quantities are obtained by measuring the cross sections at a few energy points over an energy region. This was referred to as scan experiment. The optimal design of the scan experiment (how many energy points, what the energies are, and what is the luminosity at each energy point) is of great significance both for scientific research and from economical viewpoint. Two approaches, one has recourse to the sampling technique and the other resorts to the analytical proof, are adopted to figure out the optimized scan scheme for the relevant parameters. The final results indicate that for $n$ parameters scan experiment, $n$ energy points are necessary and sufficient for optimal determination of these $n$ parameters; each optimal position can be acquired by single parameter scan (sampling method), or by analysis of auxiliary function (analytic method); the luminosity allocation among the points can be determined analytically with respect to the relative importance between parameters. By virtue of the second optimization theory established in this paper, it is feasible to accommodate the perfectly optimal scheme for any scan experiment.
\end{abstract}

\noindent
{\bf PACS} \hskip 0.25cm: 87.55.de; 87.55.kh; 13.66.Jn \\
%% 87.55.de Optimization; 87.55.kh  Monte Carlo methods Applications
%% 02.10.Yn Matrix theory 02.30.Sa Functional analysis
%% 13.66.Jn Precision measurements in e+e- interactions

\section{Introduction}

Scan method is a useful tool for various kinds of studies in domain of high energy physics.
Firstly, it plays an important role in the discovery of new resonances, the most famous one is $\jpsi$ that leads to ``November revolution'' in particle physics~\cite{Aubert:1974js,Augustin:1974xw}. Secondly, scan experiments can provide lots of accurate information related to particles, such as accurate measurement of the $\tau$ lepton mass~\cite{tau1992,tau1996}, accurate measurement of the Z resonance parameters~\cite{Barate:1999ce,Acciarri:2000ai,Abreu:2000mh,Abreu:2000mh}, and so forth.
Thirdly, scan measurement can add a lot to understand the present theory, such as R-value measurement~\cite{Bai:1999pk,Bai:2001ct} and phase angle measurement~\cite{Wang:2003zxa}, both of which are crucial for quantum chromodynamics researches.

A scan searching experiment is always intriguing and exciting. However, the scan scope is usually unpredictable large, since no one knows where the new particle will jump up. Therefore, it is fairly reasonable to show respect to pioneers for their bravery and diligence in looking for a needle in a haystack. Anyway, the situation changes into another direction step by step. Nowadays, more and more particles are discovered, and luminosity of accelerator becomes higher and higher, the scan experiments begin to play a new role in physics study, especially for the high precision measurement. The high precision measurement will help us to understand the existing theory more profound, and also helpful for the new discovery during the progress of accuracy improvement. However, since scan experiment is usually performed at many energy points, the optimal choice of energy position and luminosity distribution at each point becomes a more and more prominent issue, which directly relates to the efficiency of data taking procedure.

Scan optimization is not a trivial affair. For example, during the statistical optimization study for $\tau$ mass scan, it is found that one energy point is enough for one parameter fit~\cite{wangyk2007}. Further more, the successive studies~\cite{wangyk2009,wangbq2012,wangbq2013} indicate that for $n$ free parameters fit, $n$ scan points are enough to give optimal results. As a matter of fact, the fewer the points, the more efficiently the accelerator works, since lots of tuning time can be dispensed with. On this extent, the optimization theory that figures out the minimal number of points is of great importance for practical data taking design of scan experiment.

The theory of second optimization for scan experiment, which is depicted in following sections, will  accommodate perfect scheme for scan experiments that aim at accurate measurements of interesting parameters. This paper begins by in Sect.~\ref{sxn:notion} providing the concept of the second optimization that is the kernel of following study. The sampling method is adopted to explore optimal scan scheme in Sect.~\ref{sxn:samplingmed}, where $\tau$ mass measurement is used as a concrete example. In Sect.~\ref{sxn:theory}, the analytical theory of second optimization is established on the basis of elementary knowledge about numerical optimization. Section~\ref{sxn:discussion} devotes to some discussions involving the equivalence between likelihood and chisquare fits, optimal effect due to systematic uncertainty, correlation problem of systematic uncertainty, multiple solution issue related to the objective function, merits of the sampling method and the analytical theory. Finally, key conclusions are summarized in Sect.~\ref{sct:sumary}.

\section{Notion of Second Optimization}\label{sxn:notion}
The chisquare form for scan experiment reads
\begin{equation}
\label{eq:chisq}
\chi^2 = \sum_{i=1}^{m} \left( \frac{N^{obs}_i - N^{th}_i} {\Delta^{obs}_i} \right)^2,
\end{equation}
where $i$ denotes the $i$-th scan point, and the total number of scan points is $m$. $N$ is the number of events that is classified into two categories: the observed number of events ($N^{obs}$) and the theoretical number of events ($N^{th}$). The relation between event number ($N$), luminosity ($L$), efficiency ($\epsilon$), and cross section ($\sigma$) is expressed as
%\beq N_i=L_i \epsilon_i \sigma_i. \label{eq:nlepsilon}\eeq
\beq N=L \epsilon \sigma~. \label{eq:nlepsilon}\eeq
Generally speaking, for scan within large energy scope, the efficiency is energy dependent and distinctive at different scan point; for comparatively small scan scope, such as $\tau$ mass scan, $\jpsi$ and $\psip$ narrow resonances scan, the efficiency can treated as a constant, that is
\beq \epsilon_i =\epsilon, i=1,2,\cdots,m. \label{eq:epsilon}\eeq
Such an assumption has essentially no effect on general conclusions obtained in this paper, and is always assumed in the study that follows. $L_i$ denotes the luminosity at the $i$-th point, the relation between $L_i$ and total luminosity ($L$) is as follows
\beq L_i =x_i L, \mbox{~~with~~}\sum\limits^{m}_{i=1} x_i=1. \label{eq:lilx}\eeq
Here $x_i$ denotes the luminosity allocation at point $i$. $\Delta^{obs}$ is the error of the observed number of events. As to a Poisson distribution,
\beq \Delta^{obs}=\sqrt{N^{obs}}~, \label{eq:deltan}\eeq
the form of which is adopted in the following study.

The observed cross section can be measured through the observed number of events by relation \eref{eq:nlepsilon}. The theoretical cross section is usually acquired on the basis of present theoretical calculations that involve some parameters, which can be obtained by fitting experimental data. Mathematically,
\beq \sigma^{th} =\sigma (\theta), ~\theta = (\theta_1, \theta_2, \cdots, \theta_n)^T. \label{defthetavec}\eeq
Here $\theta$ is the parameter vector, there are totally $n$ parameters. $T$ indicates transpose of vector or matrix. In addition, the luminosity allocation vector is also introduced and defined as \beq x = (x_1, x_2, \cdots, x_m)^T. \label{defxvec}\eeq
For convenience, the observed cross section is denoted as $\bar{\sigma}$, that is $\bar{\sigma}=\sigma^{obs}$; the theoretical cross section is denoted as $\sigma$ or $\sigma(\theta)$, the latter is used to stress the dependence of cross section on parameters. In a word, $\chi^2$ can be recast as
\beq
\chi^2 (\theta,x) =L \epsilon \cdot \sum_{i=1}^{m} \frac{x_i}{\bar{\sigma}_i} \left[ \bar{\sigma}_i  - \sigma_i(\theta) \right]^2.
\label{chisqfm}\eeq
In above expression, parameters $\theta$ and $x$ denote the optimal problem we want to study. For certain $x$, the minimization of the $\chi^2$ leads to a set of optimal parameters, which are denoted as $\thetast$. This optimal process is the usual one in experimental data analysis, which is called the {\bf first optimization}. Obviously, errors of $\thetast$ depend on the values of $x$. Therefore, under the constraint of certain total luminosity or $\sum\limits^{m}_{i=1} x_i=1$, the optimization on $x$ is performed in order to obtain the smallest errors of $\thetast$. This optimization on $x$ is called the {\bf second optimization}.

In the following sections, the sampling method is utilized to study the second optimization firstly. The $\tau$ mass scan is taken as an example, due to the simplicity of which, the essence of the sampling method is exhibited pedagogically. Then, the analytical theory based on optimization principle is established, which settles the issue of scheme design for scan experiment thoroughly and perfectly.

A remark is in order here. For the $\tau$ mass scan, the conventional likelihood estimator is adopted, which is equivalent to the chisquare estimator for the first optimization (refer to subsection \ref{sxteqlkandsq}). As far as the second optimization is concerned, whichever form of estimator is chosen is actually irrelevant, since they are only relevant to the first optimization.

\section{Sampling Method}\label{sxn:samplingmed}
For the $\tau$ mass ($\taums$) scan, several points, say totally $N_{pt}$ points need to be taken in the vicinity of $\taums$ threshold. By virtue of analyzed data, the following likelihood
function is constructed~\cite{tau1992,tau1996,tau2,Ablikim:2014uzh}:
\begin{equation}
LF =\prod\limits_i^{N_{pt}} \frac{\mu_i^{N_i} e^{-\mu_i} }{N_i!}~ ,
\label{lklihd}
\end{equation}
where $N_i$ is the observed number of $\tau^{+}\tau^{-}$ events
obtained by $e \mu$-tagged final state (here the $e \mu$ channel means $\tau^+ \to e^+ \nu_e \bar{\nu}_{\tau}, \tau^- \to \mu^- \bar{\nu}_{\mu} \nu_{\tau}$, or
 $\tau^+ \to \mu^+ \nu_{\mu} \bar{\nu}_{\tau}, \tau^- \to e^- \bar{\nu}_e \nu_{\tau}$) at $i$-th scan point. Here $N_i$ is assumed obeying a Poisson
distribution, whose expectation $\mu_i$ is given by
\begin{equation}
 \mu_i(m_{\tau})=[\epsilon \cdot B_{e\mu} \cdot
\sigma^{obs} (m_{\tau},E^i_{cm})+\sigbg] \cdot {\cal L}_i~~.
\label{mudef}
\end{equation}
In Eq.~\eref{mudef}, ${L}_i$ is the luminosity at the $i$-th point; $\epsilon$ is the overall efficiency of $e\mu$ final state for identifying $\tau^{+}\tau^{-}$ events, which includes
trigger efficiency and event selection efficiency; $B_{e\mu}$ is the combined branching ratio for decays $\tau^+ \rightarrow e^+ {\nu}_e \overline{\nu}_{\tau}$ and $\tau^- \rightarrow \mu^-
\overline{\nu}_{\mu} \nu_{\tau}$, or the corresponding charge conjugate mode; $\sigma^{obs}$ (with $\taums$ as a parameter), which can be calculated by the improved Voloshin's formulas~\cite{Voloshin}, is the observed cross section measured at point $i$ with center-of-mass energy $E^i_{cm}$; and $\sigbg$ is the total cross section of background channels after $\tau^{+}\tau^{-}$ selection. If $\taums$ is set as a free parameter, the minimization of $LF$ in
Eq.~\eref{lklihd} yields the best estimation for $\taums$.

Besides $\taums$, $\epsilon$ and $\sigbg$ can be free parameters as well. The sampling technique is utilized to figure out the optimal scan scheme for one-($\taums$), two-($\taums$ and $\epsilon$ ), and three-($\taums$, $\epsilon$, and $\sigma_{BG}$) parameter fit step by step~\cite{wangyk2007,wangyk2009}.

\subsection{One parameter optimization}
Herein to achieve high precision of $m_{\tau}$ we want to find out:
\begin{enumerate}
\item What is optimal distribution (position) of data taking points;
\item How many energy points are needed for scan in the vicinity of threshold;
\item How much luminosity is required for certain precision expectation.
\end{enumerate}
In the following study concerned with statistical uncertainty, taken are efficiency $\epsilon=14.2\%$~\cite{tauxc}, energy spread $\Delta$ = 1.4 MeV~\cite{tauxc}, $B_{e\mu} = 0.06194$~\cite{pdg06}, and neglected are the corresponding uncertainties whose effects are generally
small~\cite{wangyk2007}. As to $\sigma_{BG}$, the previous experience~\cite{tau1} indicates that $\sigma_{BG} \approx 0.024$ pb which is fairly small comparing with the $\TT$ production
cross section (0.1 nb) near threshold. Moreover, for a high luminosity accelerator, a large data sample can be taken below the threshold to measure $\sigma_{BG}$ accurately. In actual fit as a constant, $\sigma_{BG}$ has tiny effect on the optimization of points distribution. Therefore, for one parameter optimization, $\sigma_{BG}$ is set to be zero, which means that the study is
background free.

In the following exploration, the value of $\taums$ itself is assumed to be known, which is set to be $m^0_\tau=1776.99~\mev$ according to PDG06~\cite{pdg06}, and under such an assumption, we attempt to answer three above questions. Nevertheless, when think twice about the first two questions, it is observed that they actually intertwist with each other, {\em i.e.} the optimal number of points depends on the distribution of points and {\em vice versa}. To resolve such a dilemma, we start from a simple distribution and find the optimal number of points, then based on which we finally determine the number of points.

\subsubsection{First searching}\label{sect_optmone}
As a tentative beginning, the energy interval to be studied
is divided evenly, viz.
\begin{equation}
E_i=E_0+(i-1)\times \delta E,\ \ \ \  (i=1,2,...,N_{pt})
\label{engdiv}
\end{equation}
where the initial point $E_0=3.545~\mbox{GeV}$, the final point $E_f=3.595~\mbox{GeV}$, and the fixed step $\delta E=(E_f-E_0)/N_{pt}$ with $N_{pt}$ being the number of energy points. For a given total luminosity (${L}$) it is also apportioned averagely at each point {\em i.e.} ${L}_i = {L}/N_{pt}$.

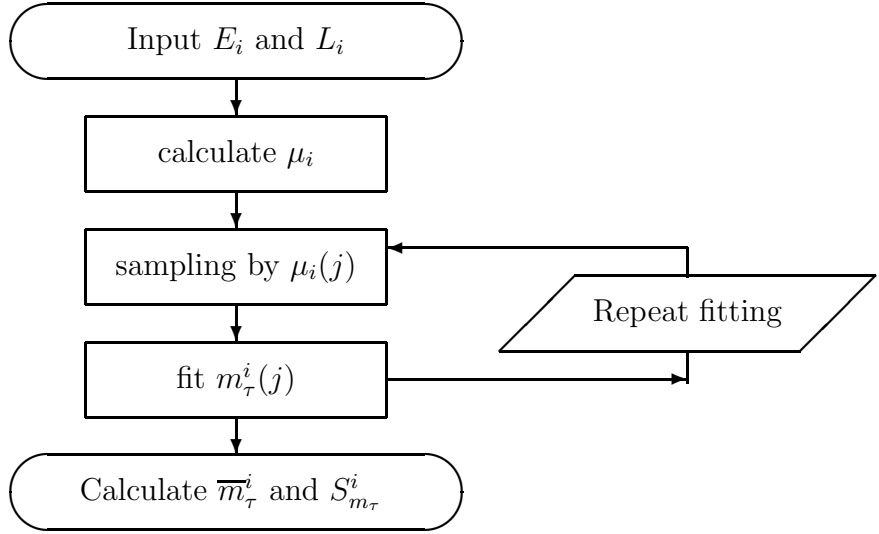
\begin{figure}[bth]
\center
\setlength{\unitlength}{0.5cm}
\begin{picture}(10,13.5)
\thicklines
\put(4.0,12.5){\oval(12.0,2)\makebox(0,0){Input $E_{i}$
    and ${L}_i$}}
\put(4.0,11.5){\vector(0,-1){1.0}}
\put(4.0,9.0){\usebox{\arrect}\makebox(0,1.){calculate $\mu_i$}}
\put(4.0,6.0){\usebox{\arrect}\makebox(0,1.){sampling by
$\mu_i(j)$}} \put(4.0,3.0){\usebox{\arrect}\makebox(0,1.){fit
$m^i_\tau(j)$}}
\put(16.0,5.25){\usebox{\arrparalla}\makebox(0,0.0){Repeat fitting}}
\put(16.0,6.2){\line(0,1){0.8}} \put(16.0,7.0){\vector(-1,0){8.0}}
\put(8.0,3.5){\vector(1,0){8.0}} \put(16.0,3.4){\line(0,1){0.8}}
\put(4.0,0.5){\oval(12.0,2)\makebox(0,0){Calculate
    $\overline{m}^i_\tau$ and $S^i_{m_\tau}$ } }
\end{picture}
\caption{\label{fig:flchart}Flow chart of sampling simulation, where $i$
($i=1,2, \cdots,N_{pt}$) indicates certain scheme and $j$
($j=1,2, \cdots,N_{samp}$) sampling times.}
\end{figure}

For each special scheme (that is for each $N_{pt}$), in order to reduce statistical fluctuation, the sampling is repeated many times (the sampling times is denoted as $N_{samp}$ ), the average value and corresponding variance of the fit out variables are worked out as
follows~\cite{Brandt} :
\begin{equation}
\overline{X}^i=\frac{1}{N_{samp}} \sum\limits_{j=1}^{N_{samp}}
X^i_{j}\ , \label{mbar}
\end{equation}
\begin{equation}
S_{X}^2(X^i)=\frac{1}{N_{samp}-1}\sum\limits_{j=1}^{N_{samp}} (X^i_{
j}-\overline{X}^i)^2 \ ~, \label{smtau}
\end{equation}
where X denotes the free fitting parameter which can be $\taums$,
$\epsilon$, and/or $\sigma_{BG}$. Here it should be noted that $i$
indicates the certain scheme, whose value can be 1 while $j$
indicates the sampling times which equals to 200 in the following
study. Without special declaration, the meaning of the average
defined by Eqs.~\eref{mbar} and \eref{smtau} will be kept in the
study that follows. The general flow chart of sampling and fitting
research is presented in Fig.~\ref{fig:flchart}.

For one-parameter optimization, $X=\taums$ and $N_{samp}=200$, using the experiment parameters given in the previous section, $\epsilon$, $\Delta$, and $B_{e\mu}$, setting ${ L}=30 \mbox{ pb}^{-1}$, and $N_{pt}$ ranging from 3 to 20, the fitted results are shown in Fig.~\ref{lmnptaradskn}(a).

\begin{figure}[htbp]
\begin{minipage}{6cm}
\includegraphics[height=5cm,width=6.cm]{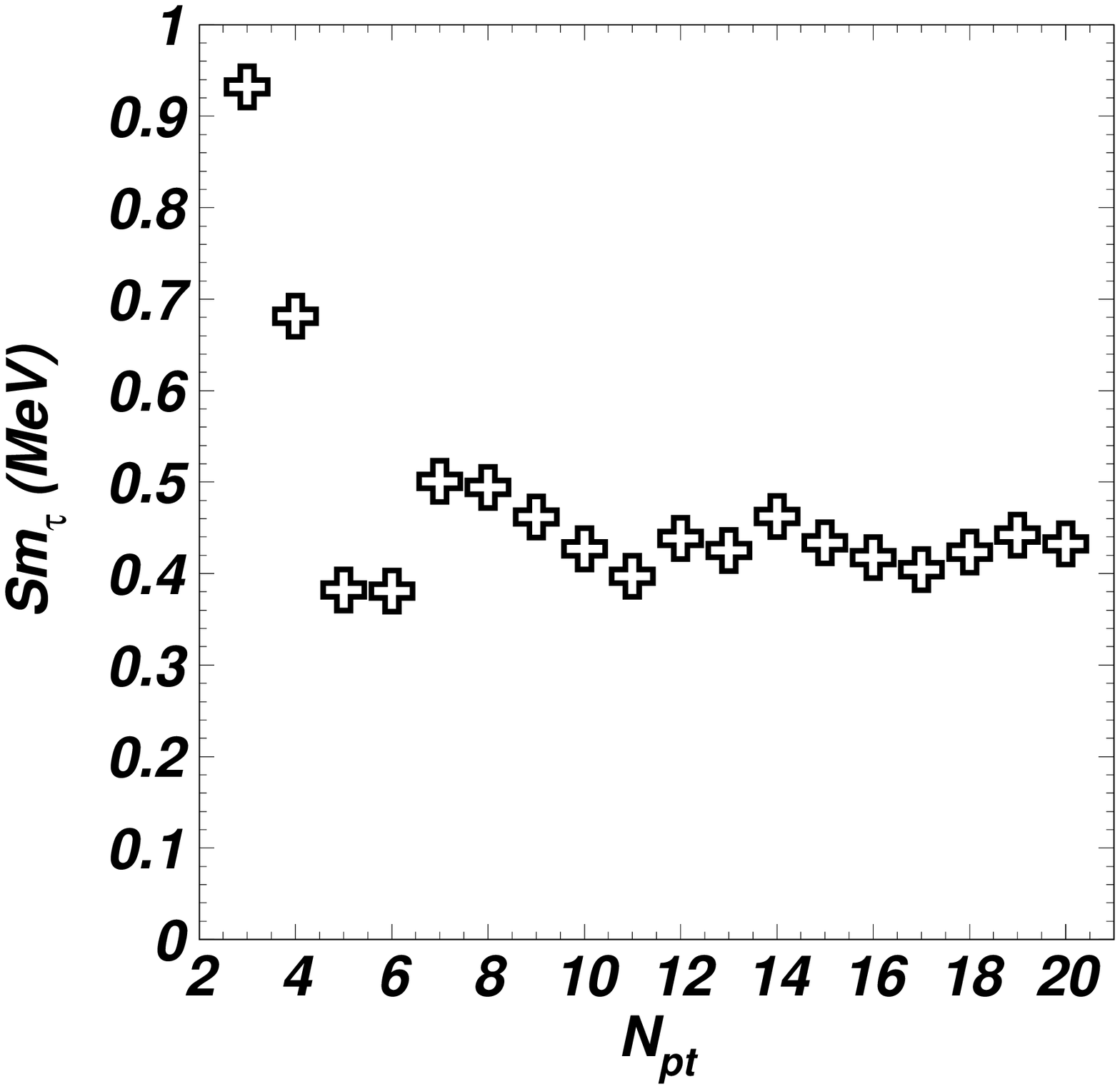}
\center (a) $S_{m_\tau}$ against $N_{pt}$
\end{minipage}
\hskip 2cm
\begin{minipage}{7cm}
\includegraphics[height=5cm,width=7.cm]{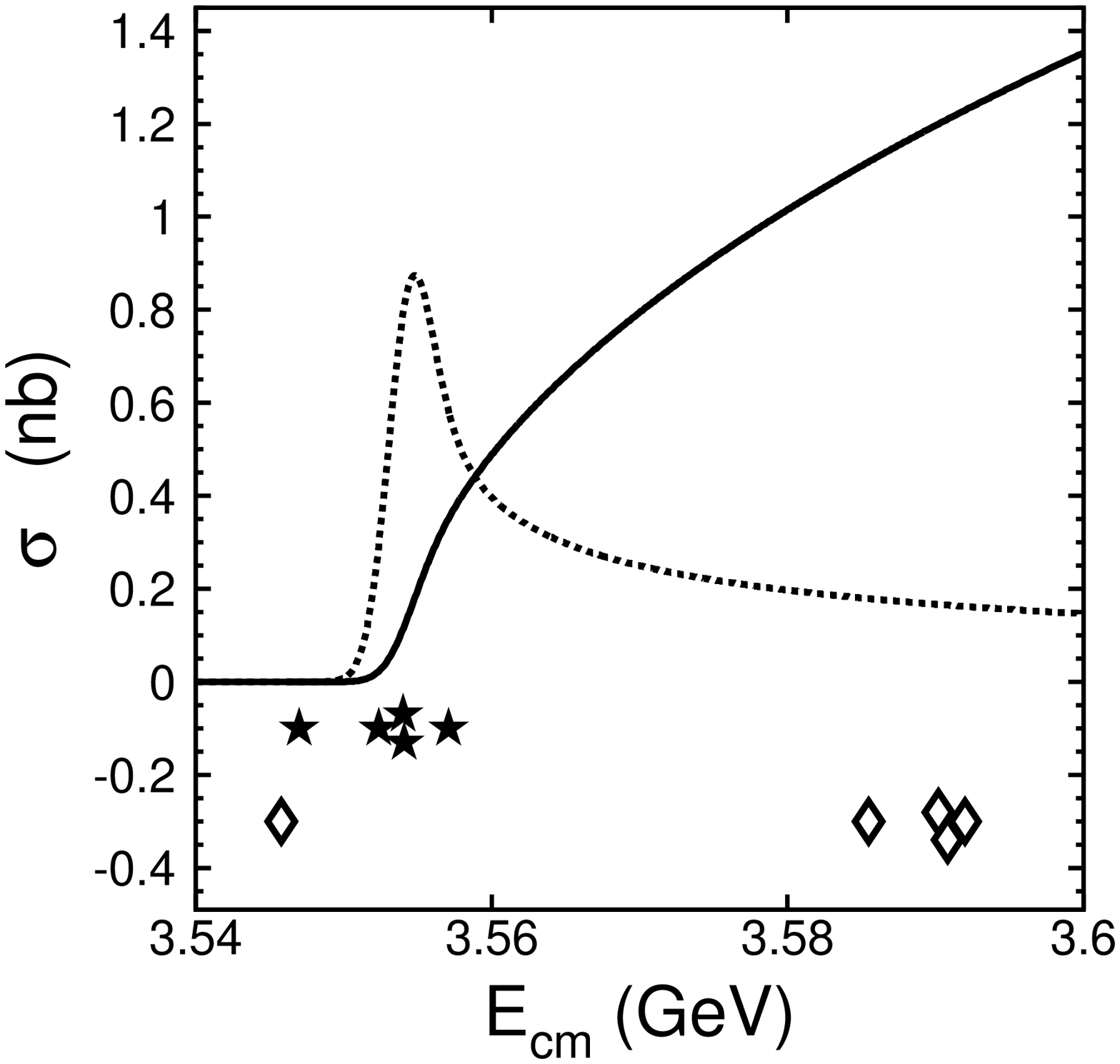}
\center (b) Two kinds of distributions
\end{minipage}
\caption{\label{lmnptaradskn}(a) shows the variation of $S_{\taums}$ against the number of points ${N}_{pt}$. (b) shows the distributions of data taking points with the smallest  and greatest $S_{\taums}$ denoted by stars and diamonds, respectively. The solid curve is the calculated observed cross section, and the dashed line the corresponding derivative of cross section to energy with a scale factor $10^{-2}$.}
\end{figure}

Here $S_{\taums}$ is the corresponding uncertainty of fitting value of $\taums$, which is adopted to assess the quality of fit, in another word, the smaller $S_{\taums}$ the better is the fit. It is prominent that too few data taking points lead to large uncertainty while too many points have no contribution for precision improvement either. As indicated in Fig.~\ref{lmnptaradskn}(a), $N_{pt}=5$ is taken as the optimized number of points for the evenly-divided-distribution scheme.

\subsubsection{Second searching}%\label{sect_opone}
With five points, we want to further search for the distribution
which can afford us the small fit uncertainty. As without any
theoretical or empirical guidance, various possibilities are
tried by employing the sampling technique, that is the
energy points is taken randomly in the chosen interval. For 200
times sampling, singled out are two fit results with the smallest
($S_{\taums}=0.152$ MeV, denoted by stars) and greatest
($S_{\taums}=1.516$ MeV, denoted by diamonds)
fit uncertainties; their distributions are shown in Fig.~\ref{lmnptaradskn}(b),
by virtue of which it is obvious when the points crowd near the
threshold the uncertainty is small; on the contrary, when the points
are far from the threshold the uncertainty becomes large.
More mathematically, it is found that the smallest uncertainty is
acquired when points gather at the region with the large derivative
of cross section to energy. So this result implies that the region
with large derivative is presumably the optimal position for data
taking. We try to prove this speculation next.

To hunt for the sensitive position, two regions are selected
as shown in Fig.~\ref{regdy}(a):
the region I ($E_{cm} \subset (3.553,~3.558)~\mbox{GeV}$ ) is selected
with the derivative falls to 75\% of its maxinum while
the region II ($E_{cm} \subset (3.565,~3.595)~\mbox{GeV}$ ) is selected
with the variation of derivative is comparatively smooth than that in region I.

\begin{figure}[htbp]
\begin{minipage}{6cm}
\includegraphics[height=5cm,width=6.cm]{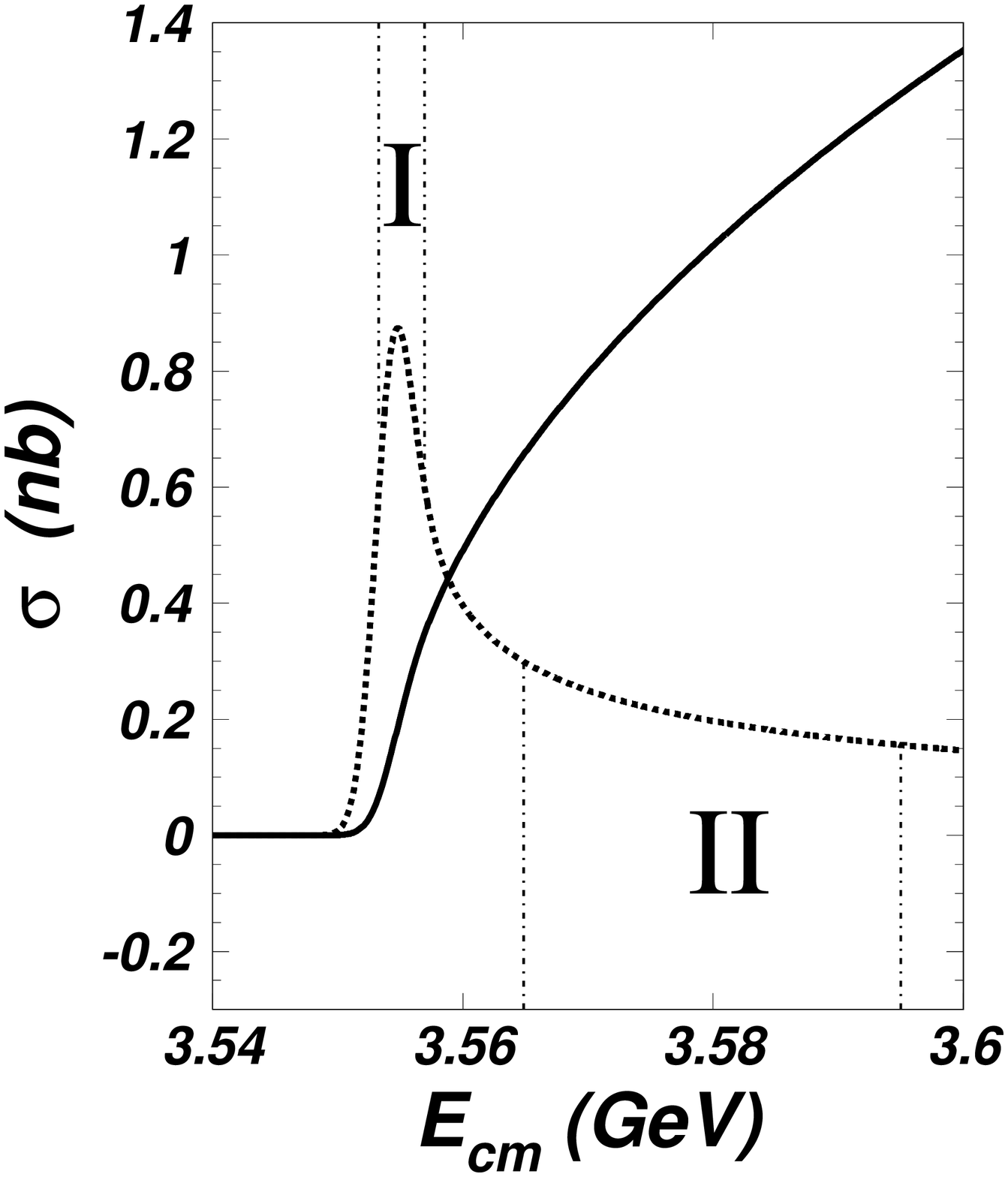}
(a) two energy regions
\end{minipage}
\hskip 2cm
\begin{minipage}{6cm}
\includegraphics[height=5cm,width=6.cm]{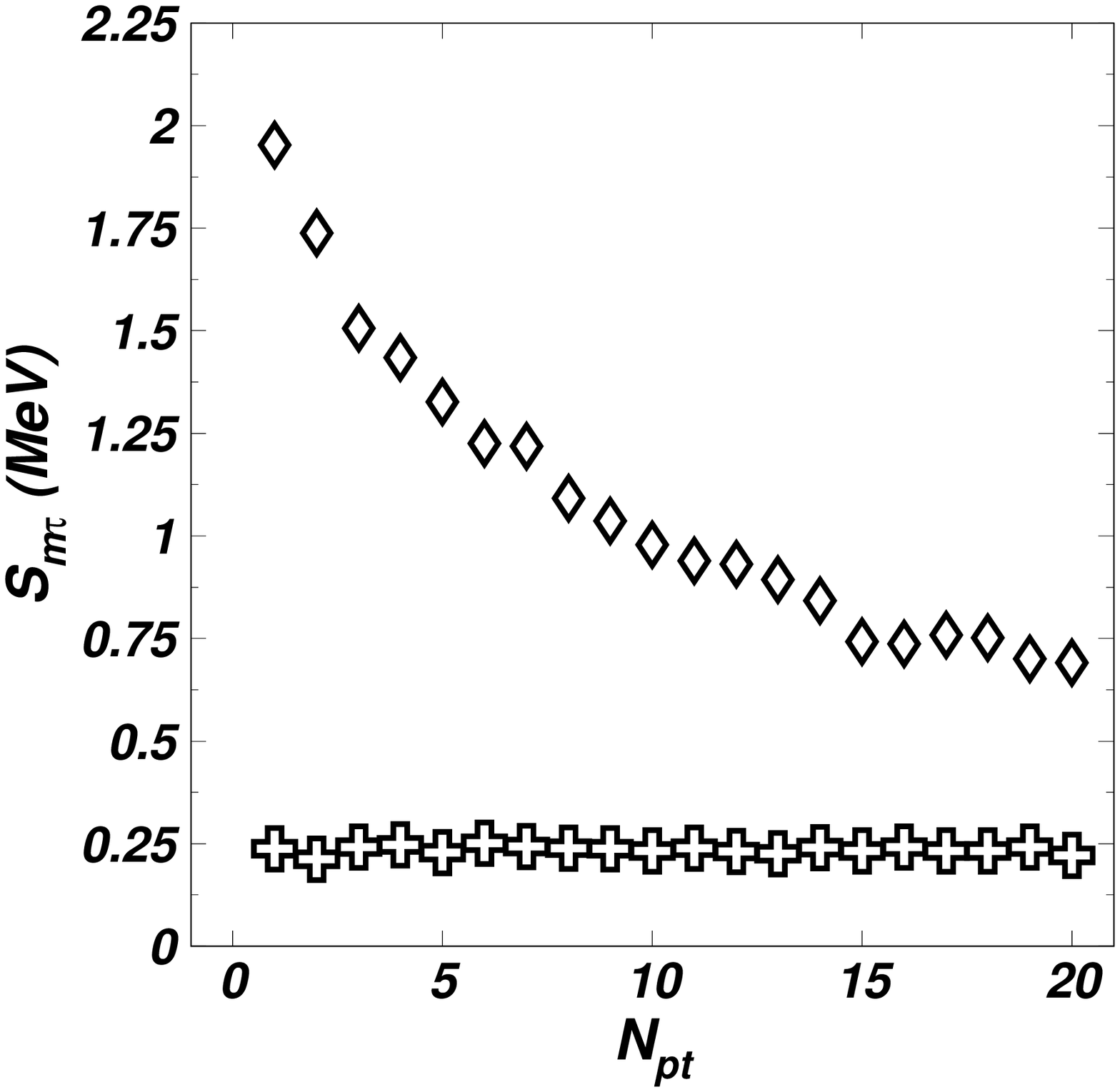}
(b) results for different scheme
\end{minipage}
\caption{\label{regdy}(a) Two subregions, denoted by $I$ and $II$,
with different derivative feature where the solid line denotes the observed
cross section and the dashed line the corresponding derivative
value with a scale factor of $10^{-2}$;
(b) the fit uncertainties for different schemes, crosses and
diamonds denote respectively the results for the first and
second schemes as depicted in the text.}
\end{figure}

To ascertain the aforementioned speculation, two schemes are designed. For the first scheme, two points are taken in the region $I$, one at 3.55398 GeV as the threshold point and
the other at 3.5548 GeV corresponding to the largest derivative point.
As in the region $II$, the number of points $N_{pt}$ increases from 1 to 20, with each point having luminosity 5 pb$^{-1}$. The fit results are displayed in Fig.~\ref{regdy}(b) by crosses. Clearly, the increase of points in the region $II$ hardly has the contribution
to accuracy improvement ($S_{\taums}=0.25$ MeV remains almost the same with the increasing number of points in region $II$). That is to say, the region $I$ is the sensitive region so far as
the fit uncertainty is concerned while the region $II$ is not. To prove this point further, for the second scheme, merely the points in the region $II$ are taken, $N_{pt}$ also increases from 1 to 20. The fit results are displayed in Fig.~\ref{regdy}(b) by diamonds. As expected, with the increasing number of points, $S_{\taums}$ decreases as well, but even with 20 points in the region
$II$ the value of $S_{\taums}=0.73$ MeV is still much larger than that with solely two points in the region $I$. Therefore it is concluded that the points within the region $I$ are more
useful for optimal data taking.

\subsubsection{Third searching}%\label{sect_opone}
In this subsection, the first thing needed to be known is how many points are optimal in the region with large derivative. As the procedure in subsection~\ref{sect_optmone}, the total luminosity
${L}=45$ pb$^{-1}$ is rationed averagely into $N_{pt}$ points ($N_{pt}=1,2, \cdots, 6$) within the energy region from 3.553 to 3.557~GeV. The results are shown in Fig.~\ref{smtauone}(a), according to
which the number of points has weak effect on the final uncertainty. In other words, within the large derivative region, one point suffices to give rise to small uncertainty. This is easy to understand since there is only one free parameter ($\taums$) needed to be fit in the $\TT$ production cross section, one measurement will further fix the normalization of the curve. The larger of the derivative, the more sensitive to the mass of the $\tau$ lepton.

\begin{figure}[htbp]
\begin{minipage}{5cm}
\includegraphics[height=5cm,width=5.cm]{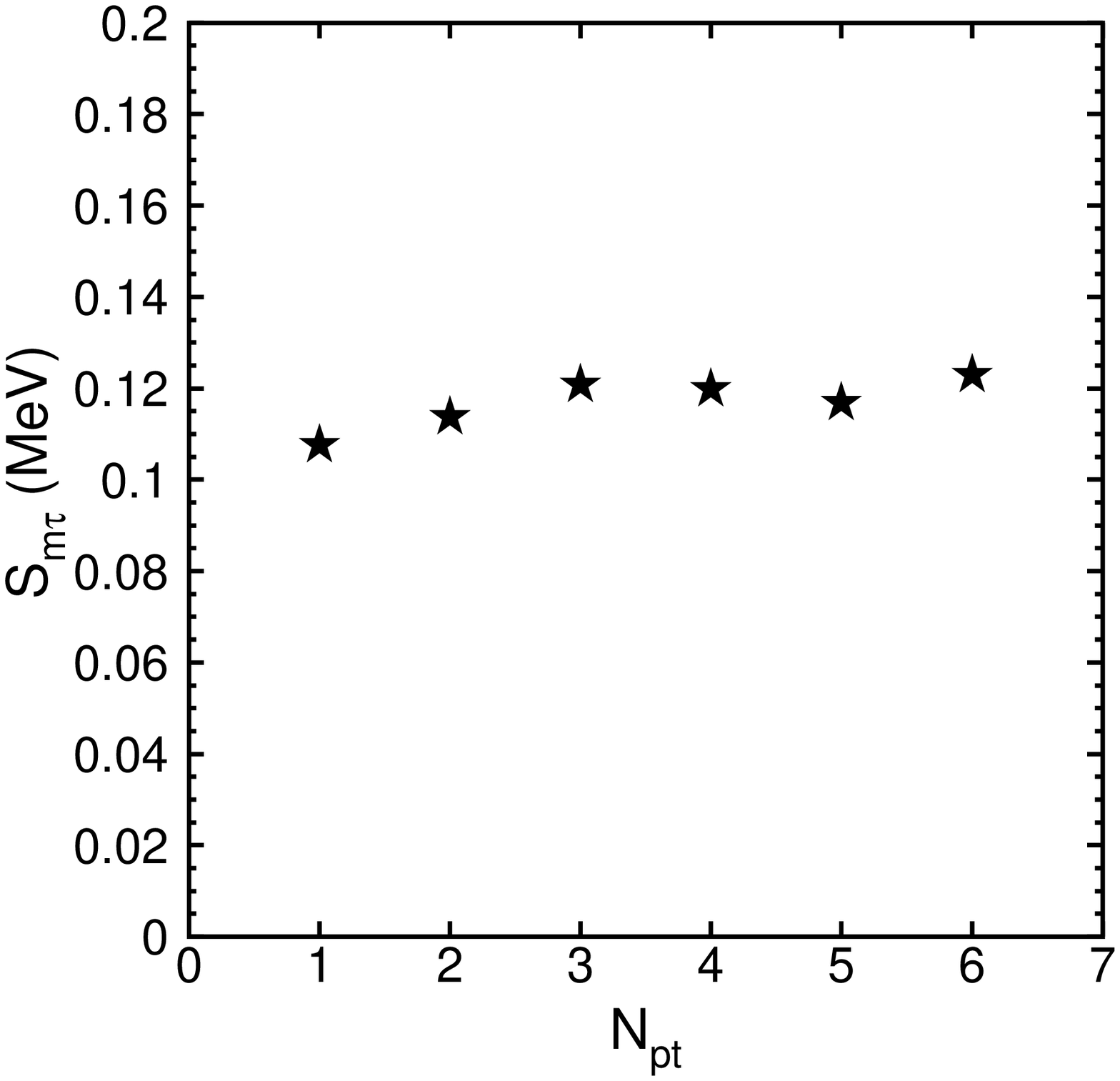}
(a) $S_{m_\tau}$ for different points
\end{minipage}
\begin{minipage}{5cm}
\includegraphics[height=5cm,width=5.cm]{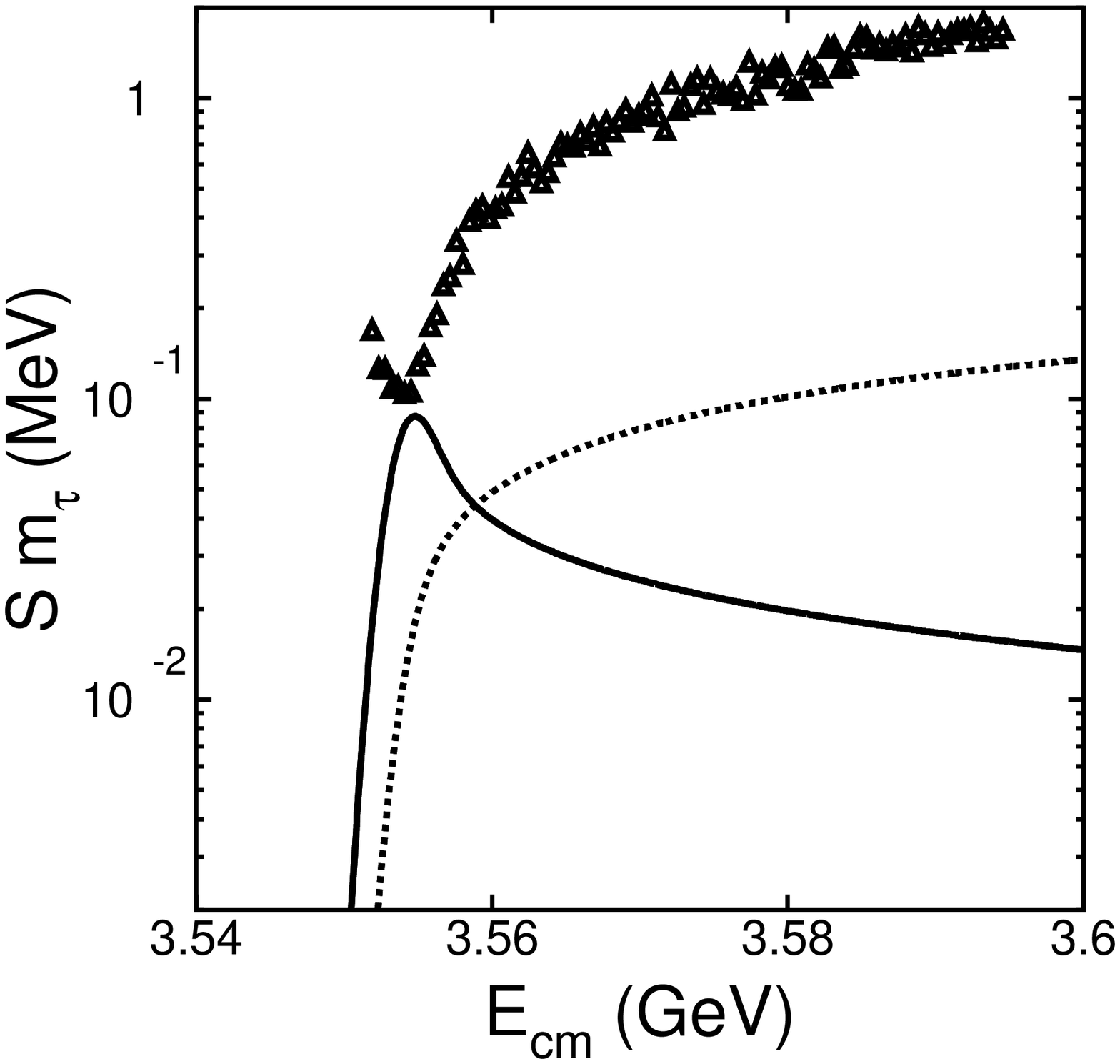}
(b) scan region from 3.551 to 3.595~GeV
\end{minipage}
%\hskip 2cm
\begin{minipage}{5cm}
\includegraphics[height=5cm,width=5.cm]{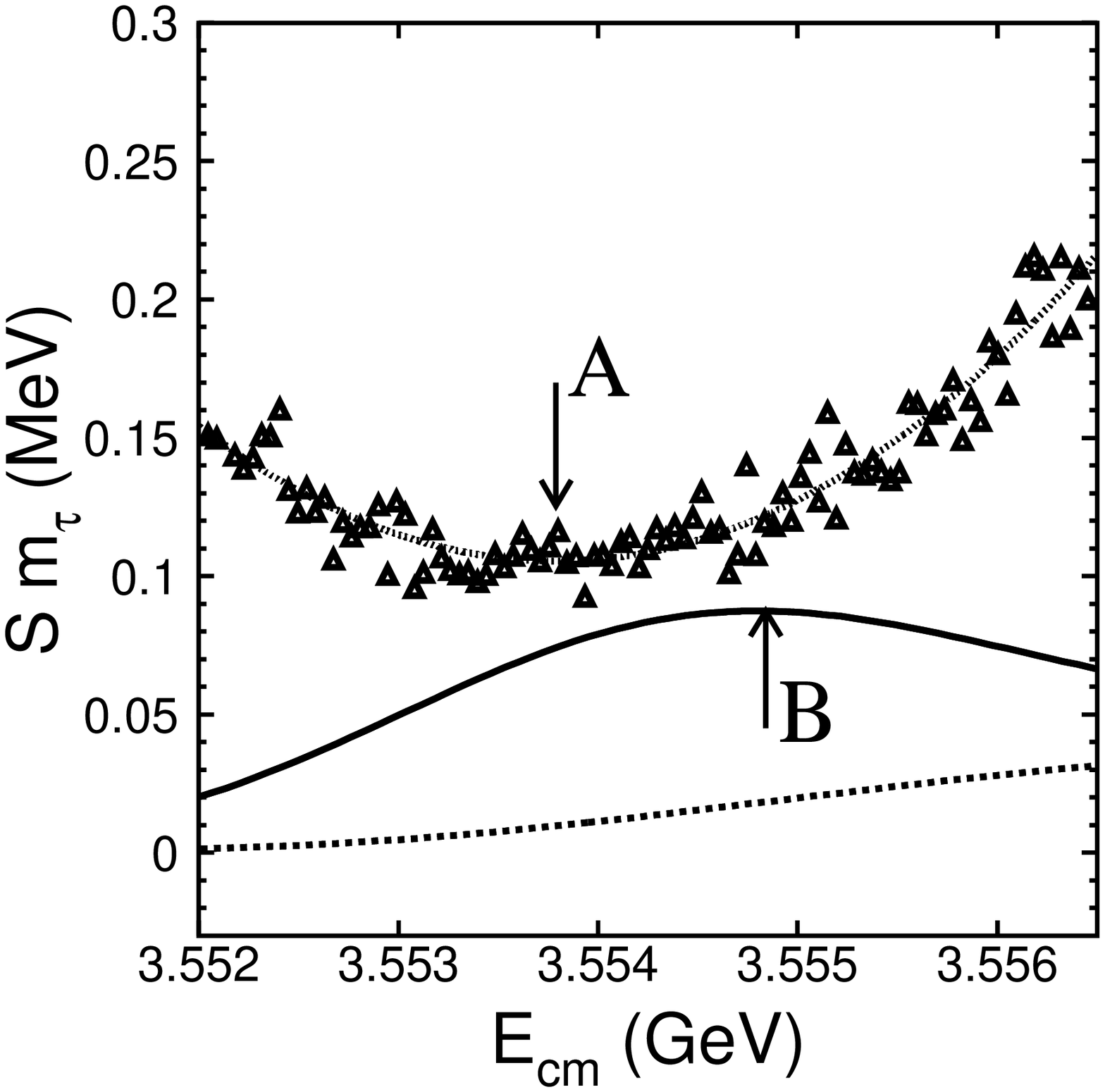}
(c) scan region from 3.552 to 3.5565~GeV
\end{minipage}
\caption{\label{smtauone}
The (a) shows the relation between $S_{m_\tau}$ and the number of points within the energy region from 3.553 to 3.555~GeV. The (b) and (c) shows the variation of $S_{\taums}$ against
energy from one point fit with ${L}=45$ pb$^{-1}$. In plot (b) the scan region is from 3.551 to 3.595~GeV while in plot (c) the scan region is from 3.55330 to 3.55694~GeV. The solid line denotes the derivative of cross section with scale factor $10^{-3}$ and the dashed line the observed
cross section with scale factor $10^{-1}$. The $A$ and $B$ denote respectively the positions with the smallest $S_{m_\tau}$ and the greatest derivative of cross section.}
\end{figure}

Since one point is enough, then an immediate question is where the optimal
point should locate? To answer it, the scan with one point with the luminosity
${L}=45$ pb$^{-1}$ is made and the results are shown in Fig.~\ref{smtauone}(b).
Just as previous study indicated, the small uncertainty is achieved near
the peak of derivative. If looking into the region from 3.5520 to 3.5565~GeV, refer to Fig.~\ref{smtauone}(c), it is found that the smallest $S_{\taums}=0.105~\mev$ is obtained
near the $\taums$ threshold ($E_{cm}=3.55398~\mev$), which has a
deviation from the position ($E_{cm}=3.55484~\mev$) with the greatest derivative of
cross section where $S_{m_\tau}=0.122~\mev$. In addition,
the study also indicates that within 2 MeV region the variation of
$S_{\taums}$ is fairly smooth (from 0.105 to 0.127~$\mev$), which is
very favorable for actual data taking.

\subsubsection{luminosity and uncertainty}%\label{sect_opone}

The empirical formula of the relation between the fit uncertainty
$S_{\taums}$ and the given total luminosity ${L}$ can be fitted based
on the data provided in Ref.~\cite{wangyk2007} as follows
\beq S_{\taums} [\kev] = \frac{708.05}{{\cal L}^{0.504}
\mbox{[pb$^{-1}$]}}~, \label{mtauerr} \eeq
which indicates that 49 pb$^{-1}$ is sufficient for a statistical precision better than 0.1
$\mev$.

\subsection{Multiple parameters optimization}
\subsubsection{Position determination}\label{sect_mltpzn}
As we already note the optimal number of point depends on the distribution of points and {\em vice versa}. Under one-parameter fit case, we employ the sampling technique to take energy points randomly in the chosen interval, which in principle exhausts all possibilities and ensure the optimization of final scheme. However, such a method is infeasible for multiple parameters fit due to the increasing complex of fit. For example, it is found when two energy points are too close to each other, the fit always fails. Before the establish of analytical theory, it is expedient to adopt the ``independence conjecture'', that is the optimization of one parameter is independent from the others. In actual operation, we fix the optimal positions which have been found, only variate one energy point for one parameter scan so that we can find the optimal position. When all optimal positions have been found, we try to investigate some possibilities to confirm the optimization of the figured out scheme.

\begin{figure}[htbp]
\begin{minipage}{5cm}
\includegraphics[height=5cm,width=5.cm]{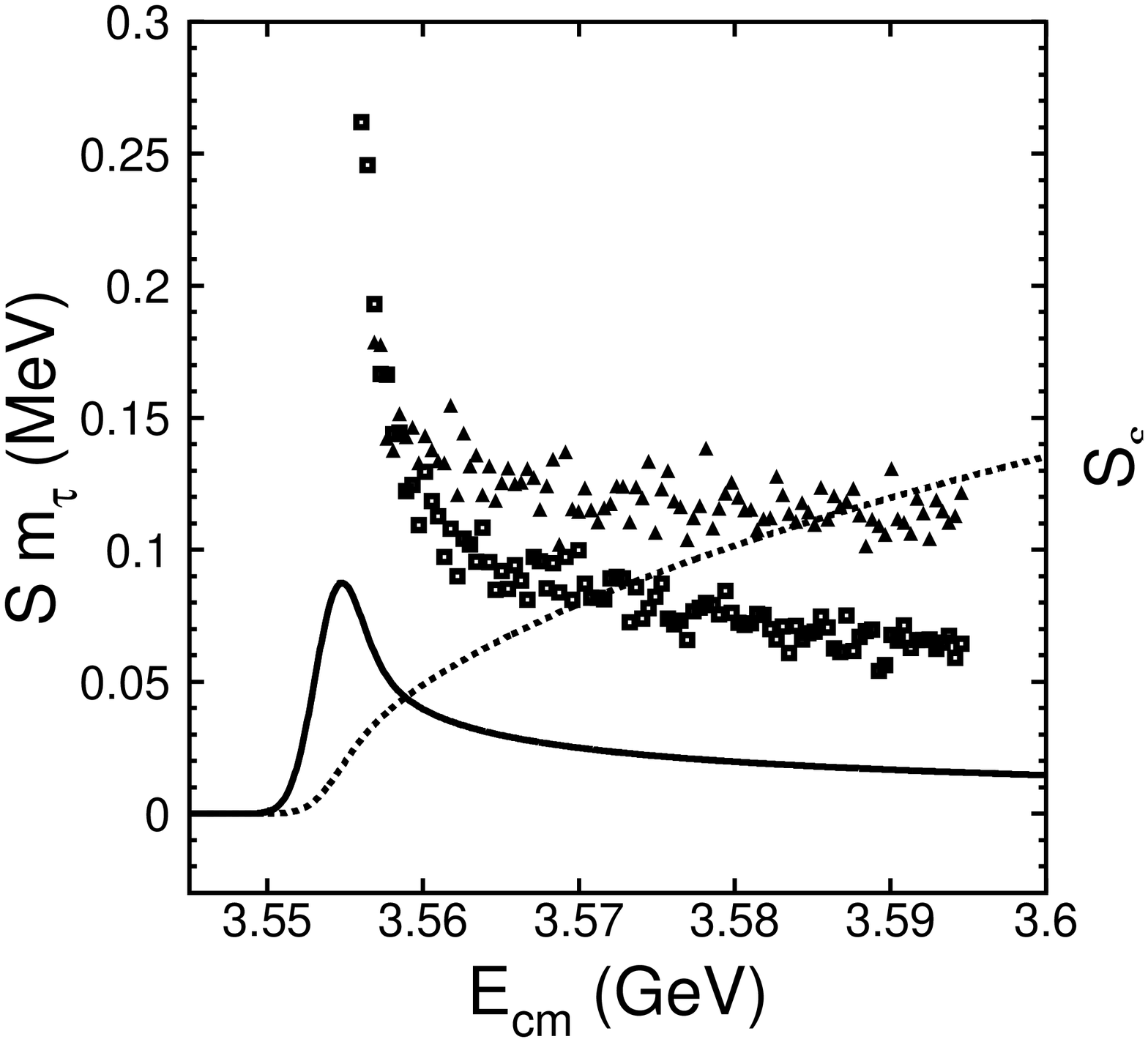}
(a) The variation of $S_{\taums}$
\end{minipage}
\begin{minipage}{5cm}
\includegraphics[height=5cm,width=5.cm]{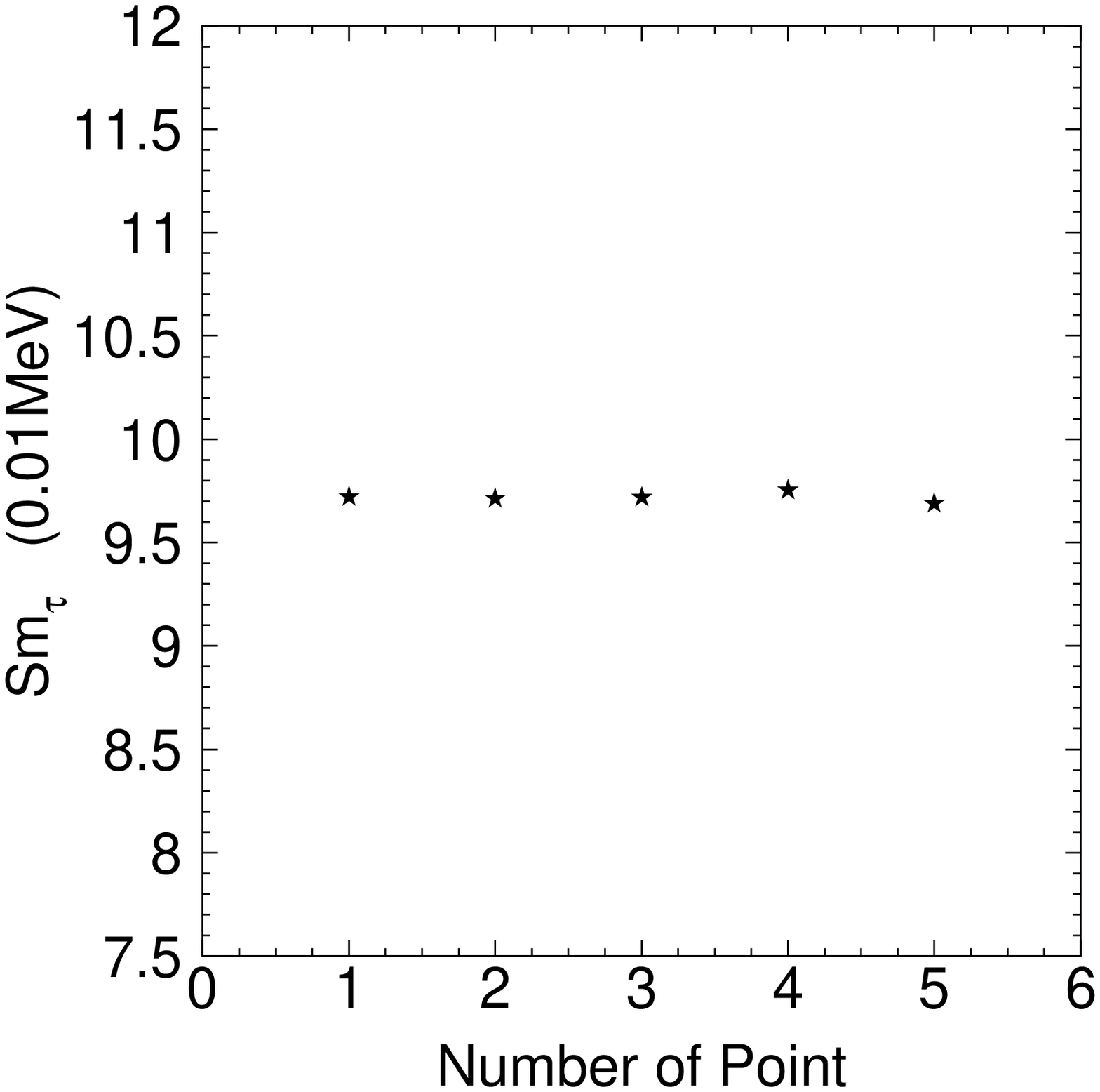}
(b)$S_{m_\tau}$ with point number
\end{minipage}
%\hskip 2cm
\begin{minipage}{5cm}
\includegraphics[height=5cm,width=5.cm]{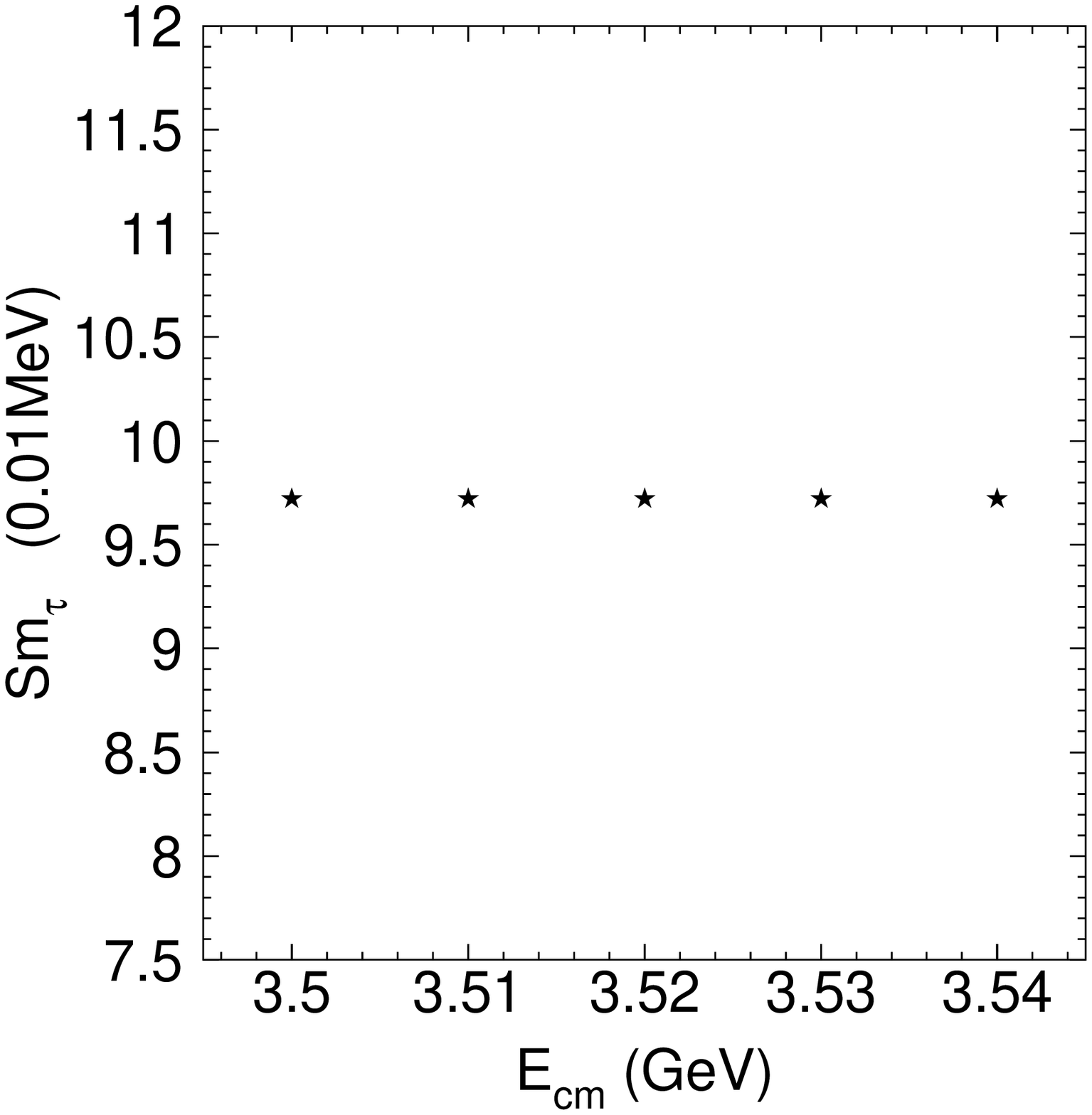}
(c) $S_{m_\tau}$ with point position
\end{minipage}
\caption{\label{smtaupp}
(a) The variation of $S_{m_\tau}$ and $S_{\epsilon}$ with the scan of the second energy point from 3.554 to 3.595~GeV. The small boxes and the small triangles represent $S_{\taums}$ and
$S_\epsilon$ respectively. The solid line denotes the derivative of cross section with a scale factor of 0.001 and the dotted line denotes the cross section with a scale factor of 0.1.
(b) The relationship between error of $\taums$ and the number of data taking points.
(c) The variation of error of $\taums$ and the location of energy.}
\end{figure}

In the light of the results of one-parameter fit, the first point is fixed at $\tau$ threshold ($E_1=3.55379~\mbox{GeV}$) to determine the parameter $\taums$. As to the new adding fit parameter
$\epsilon$, the quantity $S_\epsilon$ is used to find the optimal position for $E_2$ with the increasing energy position. Figure~\ref{smtaupp}(a) shows the distributions of
$S_\epsilon$ with the variation of second point. It is obvious that $S_\epsilon$ decreases with the decreasing of derivative. Therefore, the optimal position of the second can be
selected far from $\tau$ threshold at the high energy side, for example $E_{cm}=3.6$ ~GeV.
As a matter of fact, Figure~\ref{smtaupp}(a) also shows the uncertainty of $\taums$ remains almost the same when the upper energy point greater than 3.58 GeV, which means the upper-limit of the scan is not crucial for $\taums$ determination, that is to say, the upper-limit could be selected with large freedom, such as 3.59, 3.595, 3.6, 3.605 GeV, and so on. Anyway, when the energy is greater than 3.65 GeV the effect due to the resonance of $\psi(3686)$ will exhibit~\cite{Wang2003}. Therefore, the upper-limit of $\taums$ scan should be less than 3.65 GeV.

The conclusions listed here are easy to be understood. First, for one parameter, since
there is only one free parameter ($\taums$) needed to be fit in the $\TT$ production cross section, one measurement will fix the shape of the curve. Second, the fitted parameter will be sensitive
to the variation of curve. Mathematically, the variation of curve is reflected by its derivative. So the sensitive point for $\taums$ will be in the region with large derivative.
As shown in Fig.~\ref{regdy}, two regions are selected: the region I ($E_{cm}
\subset (3.553,~3.558)~\mbox{GeV}$ ) is selected with the derivative falls to 75\% of its maximum while the region II ($E_{cm} \subset (3.565,~3.595)~\mbox{GeV}$ ) is selected with the variation of
derivative is comparatively smooth than that in region I. In region I, the variation of derivative against the energy is fairly prominent which indicates such a region will be sensitive to the horizontal change (that is energy scale change). Therefore, region I is optimal
for $m_{\tau}$ which is determined by both the shape of the cross section curve and the energy scale. Comparatively, the variation of derivative in region II is smooth so it is insensitive to the
horizontal change but can be sensitive to the vertical change. That is to say, it could be expected that region II will be optimal for efficiency which determines the overall normalization of the curve. This is just the results displayed from the scan of $S_\epsilon$.

Based on results of the preceding section, two parameters $\taums$ and $\epsilon$ can be determined by the optimized first and second points which are located respectively at $E_1=3.55379~\mbox{GeV}$
%and $E_2=3.595~\mbox{GeV}$ with ratio of luminosity between two points fixing at 3 to 1.
and $E_2=3.6~\mbox{GeV}$. As to the new adding fit parameter $\sigbg$, we divide luminosity 20~$\mbox{pb}^{-1}$ into 1, 2, 3, 4 or 5 point/points within the range from 3.50 to 3.54~ GeV (the
luminosities for point 1 and 2 are $L_{1}=75~\mbox{pb}^{-1}$ and $L_{2}=25~\mbox{pb}^{-1}$, the reason for such a division refer to the next section), %Sect.~\ref{sect_ratioplm}),
fit results are shown in Fig.~\ref{smtaupp}(b). It can be seen that the number of points
has almost no effect on the fit uncertainty of $\taums$ or in another word, one point (denoted as point 3 hereafter) is enough to determine the parameter $\sigbg$. As the second step, with the luminosity of $20~\mbox{pb}^{-1}$ for the third point, we perform the fit with $E_3=3.50, 3.51, 3.52, 3.53$, or 3.54~ GeV, respectively. The relation between $S_{\taums}$ and the energy position is shown in Fig.~\ref{smtaupp}(c) which indicates that $S_{\taums}$ is almost irrelevant to energy position, as long as it is below $\TT$ threshold. This is also understandable since the cross section below threshold is always zero. Then, any position below threshold is feasible for $\sigbg$ determination. As an example, $E_{3} =3.50$~GeV is chosen as the third point.

\subsubsection{Ratio determination}\label{sect_ratioplm}
Unlike one-parameter fit, besides finding the relation between luminosity and precision, it is also necessary to know the luminosity allocation among different points. As the first step, we begin from two parameters case. For certain total luminosity, say $L=120$ pb$^{-1}$,
distinctive allocation schemes are checked and results are displayed
in Fig.~\ref{ratiosofpar}(a). Just as expected, with the increasing of
$L_1$ (decreasing of $L_2$), $S_{\taums}$ ($S_\epsilon$) decreases
(increases) correspondingly. The abnormal increasing of $S_{\taums}$
($S_\epsilon$) at extreme region where $L_2$ ($L_1$) is almost zero,
can be explained as the correlation effect between $S_{\taums}$ and
$S_\epsilon$. By virtue of the curve from fitting the data in
Fig.~\ref{ratiosofpar}(a), the minimal value of $S_{\taums}$ is achieved with $L_{1}=75$ pb$^{-1}$ or equivalently $L_{1}:L_{2}=3:1$. Further checks indicate such a ratio is independent on the total luminosity~\cite{wangyk2009}.

\begin{figure}[htbp]
\begin{minipage}{5cm}
\includegraphics[height=5cm,width=5.cm]{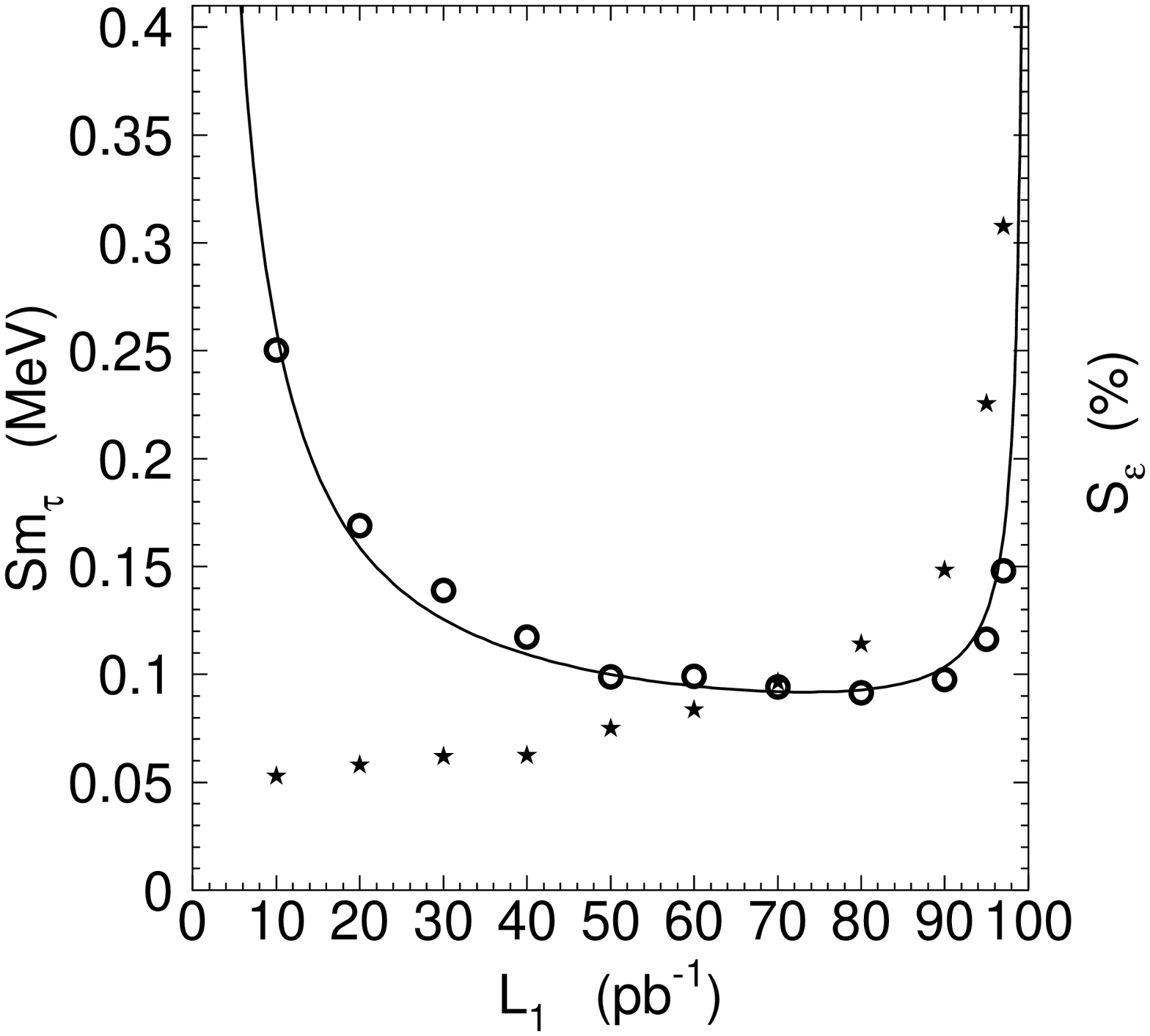}
(a) $L_2/L_1$
\end{minipage}
\begin{minipage}{5cm}
\includegraphics[height=5cm,width=5.cm]{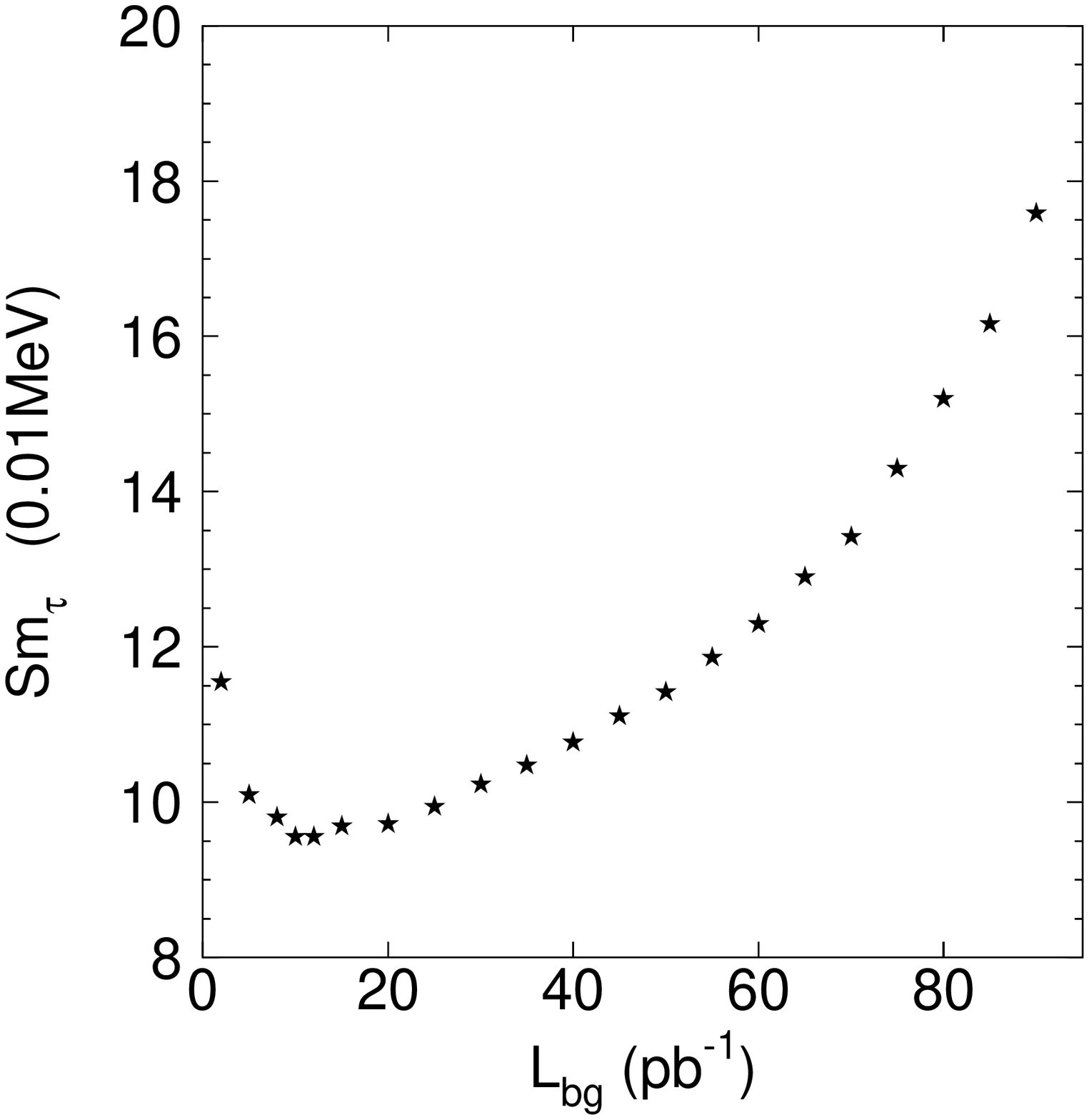}
(b)$L_3/L$
\end{minipage}
%\hskip 2cm
\begin{minipage}{5cm}
\includegraphics[height=5cm,width=5.cm]{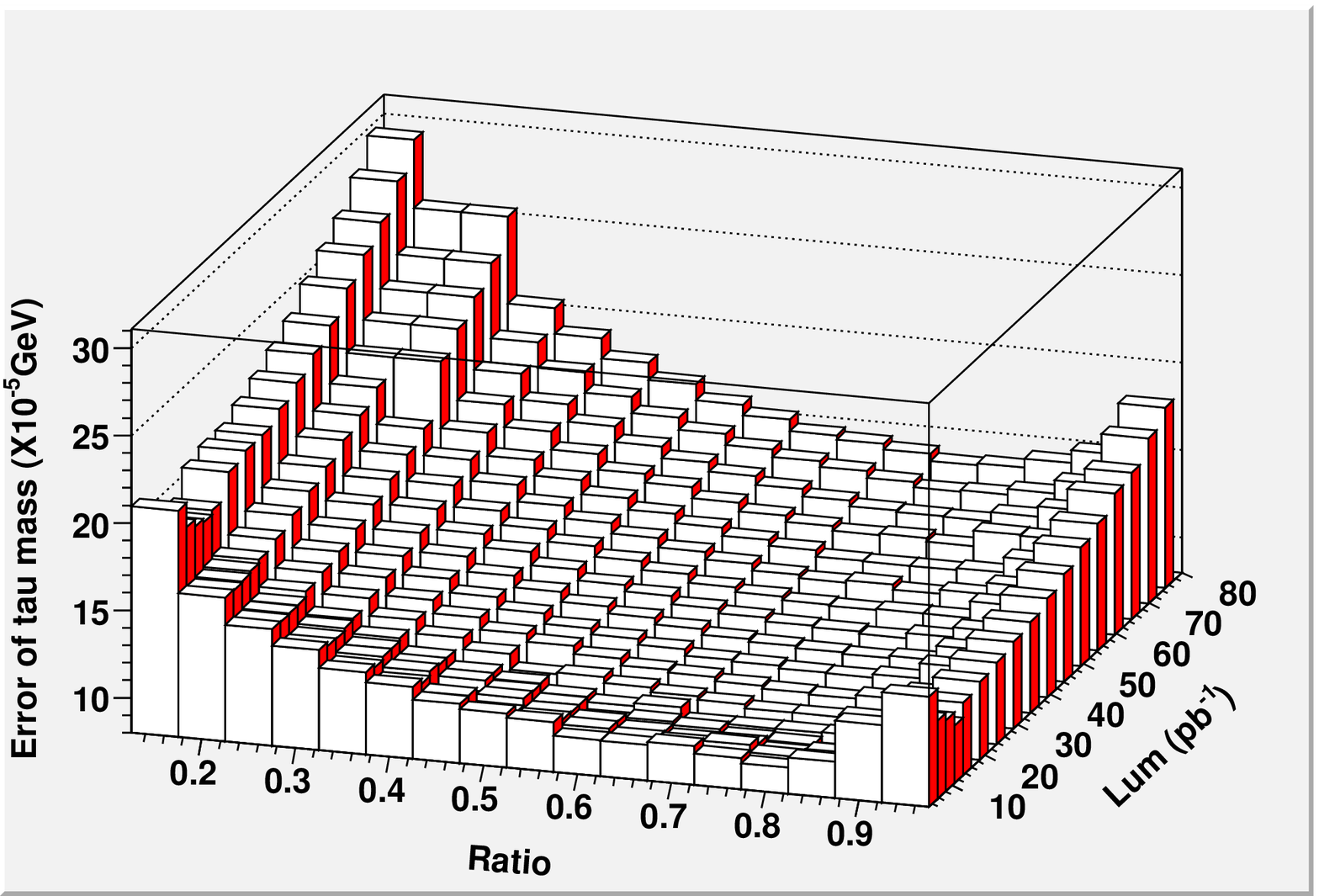}
(c) $S_{m_\tau}$ with point position
\end{minipage}
\caption{\label{ratiosofpar}
(a) The variations of $S_{\taums}$ and $S_\epsilon$ with the increasing of $L_1$.
(b) The relationship between the statistic error of $\taums$ and the luminosity at the background
point. The total luminosity is fixed as 120 $\mbox{pb}^{-1}$, the ratio of apportion luminosity at the first and the second point is 3 to 1.
(c) The variation $S_{\taums}$ with a two dimension scan of luminosity on the third point and the ratio of the luminosity allotted at the first and second points. The total luminosity is fixed at 120 pb$^{-1}$.}
\end{figure}

As the second step, we fix the total luminosity as 120~$\mbox{pb}^{-1}$ and the ratio of luminosities between the first and the second points as $3:1$, then increase the proportion of the luminosity
allotted at the third point to find the dependence of $S_{\taums}$ on the luminosity of point 3. As shown in Fig.~\ref{ratiosofpar}(b), the smallest error $S_{\taums}$ = 0.096~ MeV is obtained when the luminosity equals 12~$\mbox{pb}^{-1}$, which is about 10\% of the total.
%For the third point just determines the background cross section, 10\% of the total luminosity is reasonable. That is to say $L_3 =10\% \cdot L$ together with $L_{1}:L_{2}=3:1$ will lead to the optimal value of $S_{\taums}$.
In a word, $L_3 =10\% \cdot L$ together with $L_{1}:L_{2}=3:1$ will lead to the optimal value of $S_{\taums}$.

Some checks are performed to consolidate the obtained optimal scheme~\cite{wangyk2009}, one of which is shown in Fig.~\ref{ratiosofpar}(c), where for the fixed total luminosity, say $L = 120$ pb$^{-1}$, with $L =L_1 +L_2 +L_3$, a two dimension scan of $S_{\taums}$ is performed with respect to $L_1/(L_1+L_2)$ and $L_3$. Clearly, for the fixed $L_3$, the smallest $S_{\taums}$ is obtained at the value $L_1/(L_1+L_2)=0.75$ while for this fixed ratio, the smallest $S_{\taums}$ is obtained at the value $L_3 =12$ pb$^{-1}$, which is 10\% of the total luminosity. In fact, the smallest $S_{\taums}$ can be read directly from the three-dimension plot, with the coordinates $L_1/(L_1+L_2) \approx 0.75$  and $L_3 \approx 10\% \cdot L$.

\subsubsection{Scan scheme}\label{sect_scanskm}

\begin{figure}[htbp]
\begin{minipage}{7.5cm}
\includegraphics[height=5cm,width=5.cm]{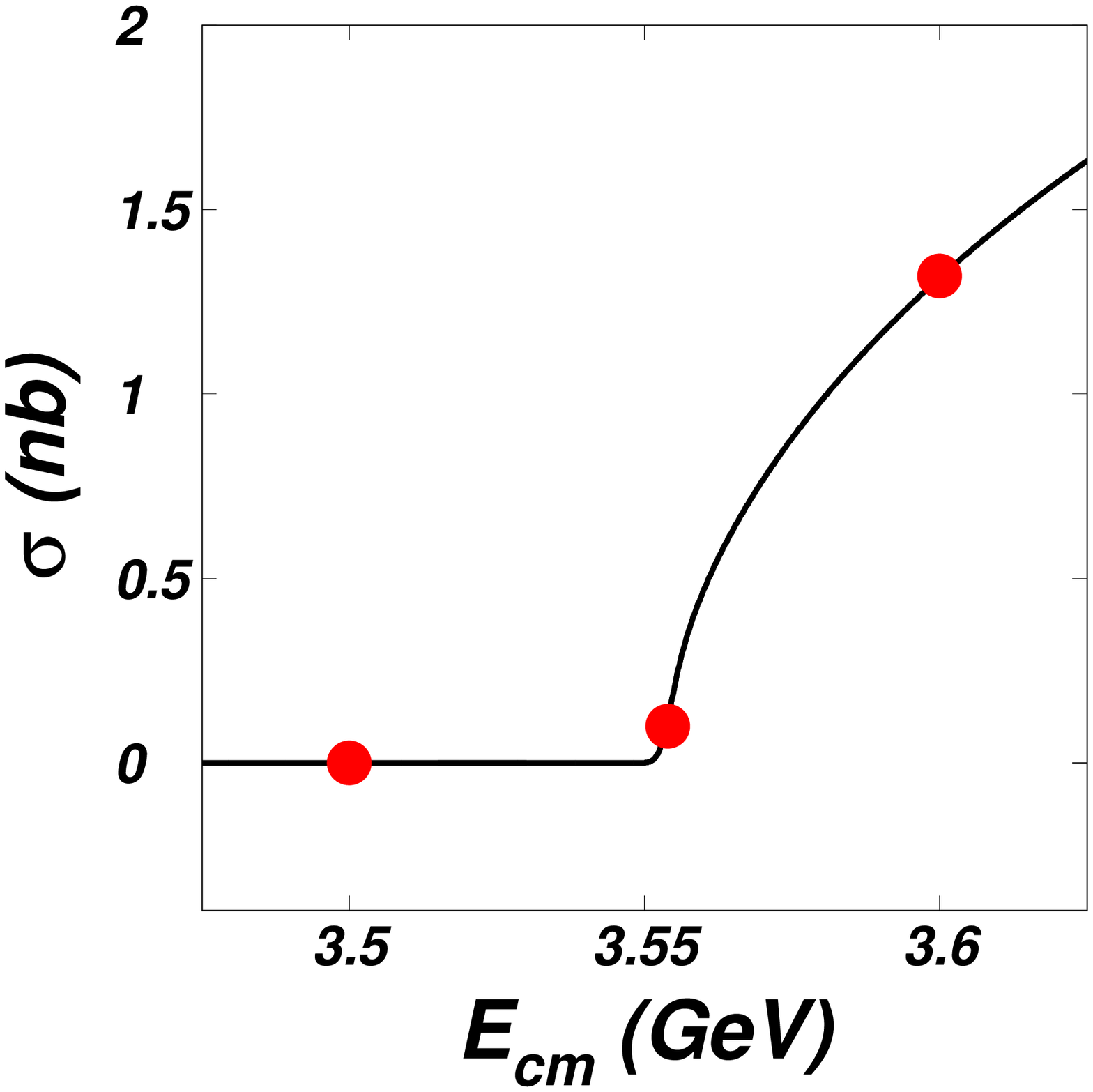}
\end{minipage}
\begin{minipage}{7.5cm}
\includegraphics[height=5cm,width=5.cm]{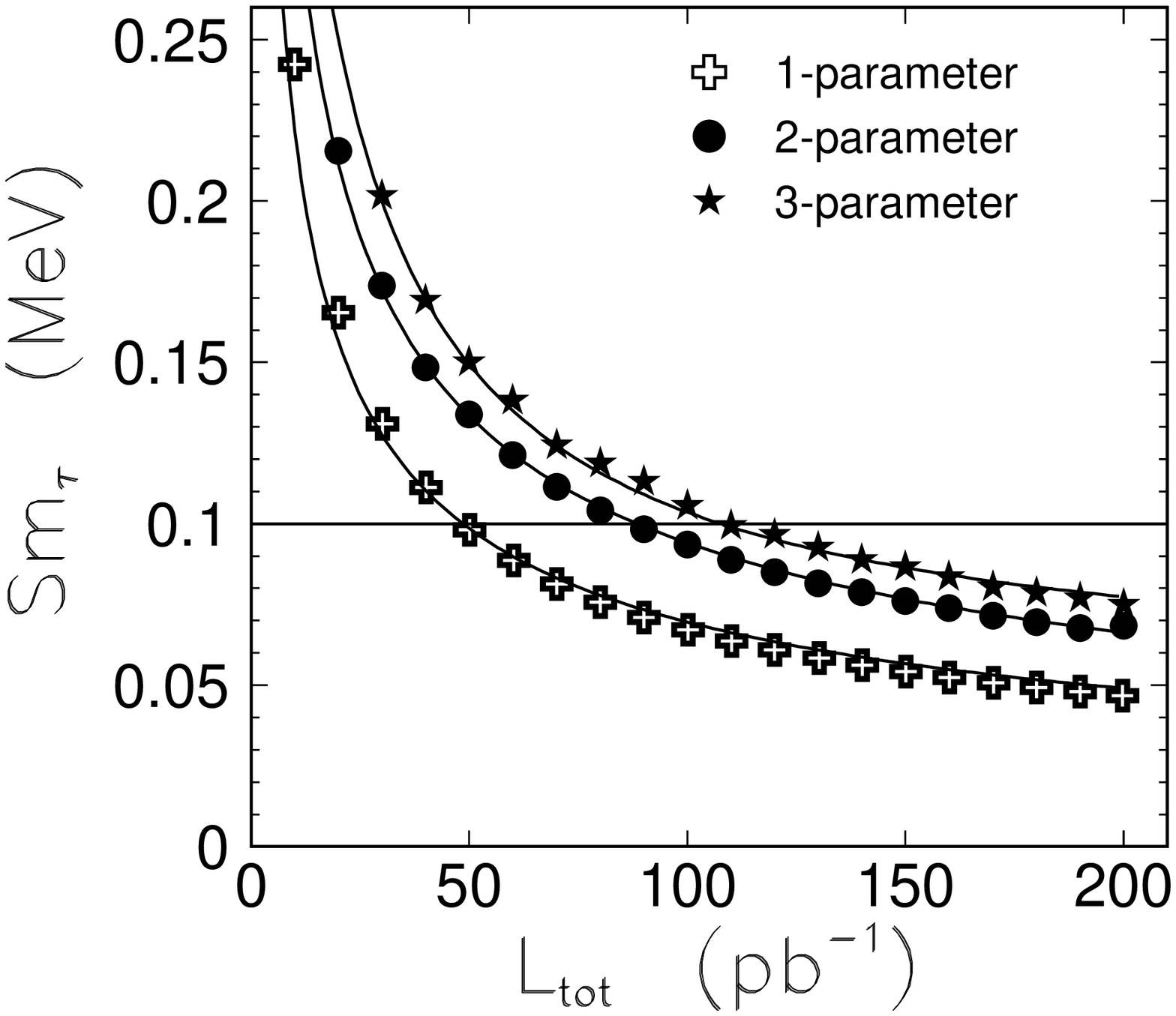}
\end{minipage}
\caption{\label{fig:resultsum}Left penal: the optimal points for $\taums$ scan. Three solid circles denote the sensitive positions for parameter 3 ($\sigbg$), 1 ($\taums$), and 2 ($\epsilon$) respectively. Right panel: the dependence of $S_{\taums}$ on $L$ for one-, two- and three-parameter fit strategies. Overlaid are fits of functions with the form $A/L^{0.504}$, with $A$ is a constant from fit.}
\end{figure}

Resorting to the sampling technique, we finally fix the optimal scan scheme for highly accurate $\taums$ measurement. For three parameters fit strategy, the positions of three data taking points are shown in the left panel of Fig.~\ref{fig:resultsum}, where three solid circles denote the sensitive positions for parameter 3 ($\sigbg$), 1 ($\taums$), and 2 ($\epsilon$) respectively. The luminosity allocation for these points are as follows: $L_{3}/L=10\%$, $L_{1}:L_{2}=3:1$, which can lead to optimal fit results for $\taums$ measurement. The relation between total luminosity ($L$) and the uncertainty of $\taums$ ($S_{\taums}$) is shown in the right panel of Fig.~\ref{fig:resultsum}, from which it can be seen that if the precision at the level of 0.1 MeV is required, around 120 pb$^{-1}$ data have to be taken. The actual experiment at BESIII just follows such a scheme~\cite{Ablikim:2014uzh,ivan2012}.

\section{Analytical Theory}\label{sxn:theory}
Now we return to Eq.~\eref{chisqfm}. Introducing
\beq\label{eq:ridef}
r_i(\theta)=\sqrt{\frac{L\epsilon x_i}{\bar{\sigma}_i}} [\bar{\sigma}_i -\sigma(\theta)] , \eeq
then the objective function $f$ is constructed as follows
\beq f(\theta) =r(\theta)^T r(\theta) = \sum\limits^{m}_{i=1} [r_i(\theta)]^2. \eeq
Obviously, $f(\theta) = \chi^2 (\theta,x)$, symbol $x$ is usually suppressed when only considered is the first optimization. The first and second derivatives of $f$ are expressed as follows
\beq \begin{array}{rcl}
\nabla f(\theta) &\equiv & 2 h (\theta)  \mbox{~~[ gradient of $f(\theta)$]~~} \\
                &=& 2 \sum\limits^{m}_{i=1} r_i(\theta) \nabla r_i(\theta) \\
                &=& 2 J(\theta)^T  r(\theta)~;
\end{array} \label{eq:fgradient}\eeq
\beq \begin{array}{rcl}
\nabla^2 f(\theta) &\equiv & G(\theta) \mbox{~~[ Hesse matrix of $f(\theta)$]~~} \\
 &=&2\{\sum\limits^{m}_{i=1} \nabla r_i(\theta)\nabla r_i(\theta)^T
                            +r_i(\theta)\nabla^2 r_i(\theta)\}\\
 &=&2\{\sum\limits^{m}_{i=1} J(\theta)^T J(\theta) + S(\theta) \}~;
\end{array} \label{eq:fhessian}\eeq
where
\beq
J(\theta)=\left(
\begin{array}{cccc}
{\displaystyle \frac{\partial r_1}{\partial \theta_1}}&
{\displaystyle \frac{\partial r_1}{\partial \theta_2}}&\cdots &
{\displaystyle \frac{\partial r_1}{\partial \theta_n}}\\
{\displaystyle \frac{\partial r_2}{\partial \theta_1}}&
{\displaystyle \frac{\partial r_2}{\partial \theta_2}}&\cdots &
{\displaystyle \frac{\partial r_2}{\partial \theta_n}}\\
\vdots&\vdots&\ddots&\vdots \\
{\displaystyle \frac{\partial r_m}{\partial \theta_1}}&
{\displaystyle \frac{\partial r_m}{\partial \theta_2}}&\cdots &
{\displaystyle \frac{\partial r_m}{\partial \theta_n}}
%\frac{\partial r_1}{\partial \theta_1}&\frac{\partial r_1}{\partial \theta_2}&\cdots &
%\frac{\partial r_1}{\partial \theta_n} \\
%\frac{\partial r_2}{\partial \theta_1}&\frac{\partial r_2}{\partial \theta_2}&\cdots &
%\frac{\partial r_2}{\partial \theta_n} \\
%\vdots&\vdots&\ddots&\vdots \\
%\frac{\partial r_m}{\partial \theta_1}&\frac{\partial r_m}{\partial \theta_2}&\cdots &
%\frac{\partial r_m}{\partial \theta_n}
\end{array}
\right),
\label{jdef}\eeq
which is an $m \times n$ matrix; and $S(\theta)\equiv r_i(\theta)\nabla^2 r_i(\theta)$, which is simply the normalized residual~\cite{Hughes}. Statistically, the normalized residuals should be small,
and scattered randomly around zero in the vicinity of minimization point of $f(\theta)$; hence, on being summed, these terms yield a negligible contribution to the Hessian $G(\theta)$. So we introduce the reduced Hesse matrix that is defined as
\beq H(\theta) \equiv J(\theta)^T J(\theta).\label{defrdchesmtx}\eeq
Next we adopt Gauss-Newton Algorithm to obtain optimal parameters :\\
\noindent {\bf Gauss-Newton Algorithm}
\begin{enumerate}
\item Given initial parameter $\theta_0$, assign the precision $\omega_1$, $\omega_2$, and set $k=0$;
\item Compute  $r_k=r(\theta_k), f_k=f(\theta_k)$;
\item Compute
\[\begin{array}{rcl} h_k &=& h(\theta_k) = J(\theta_k)^T  r(\theta_k), \\
                     H_k &=& H(\theta_k) = J(\theta_k)^T  J(\theta_k);
\end{array} \]
\item Compute $\theta_{k+1} = \theta_k -H_k^{-1} h_k$;
\item Compute $r_{k+1}=r(\theta_{k+1}), f_{k+1}=f(\theta_{k+1})$;
\item Check H-criterion, if it is satisfied, output $r_{k+1}, f_{k+1}$, stop; otherwise, set $k=k+1$, $r_{k}=r_{k+1}$, $f_{k}=f_{k+1}$, and go to step 3.
\end{enumerate}
The H-criterion is the so-called {\it Himmelblau's convergence criterion} that is defined as follows: it is assumed that $\theta_k, \theta_{k+1}, f_{k}$, and $f_{k+1}$ are computed, $\omega_1$ and $\omega_2$ are given precisions, if $\|\theta_k\| \leq \omega_1$ and $|f_{k}|\leq \omega_1$, then use $\|\theta_{k+1}-\theta_k\| < \omega_2$ and $|f_{k+1}-f_{k}| < \omega_2$ as convergence criteria; if $\|\theta_k\| > \omega_1$ and $|f_{k}| > \omega_1$, then take ${\displaystyle \frac{\|\theta_{k+1}-\theta_k\|}{\|\theta_k\|} < \omega_2}$ and ${\displaystyle \frac{|f_{k+1}-f_{k}|}{|f_{k}|} < \omega_2}$ as convergence criteria.
Herein, it is worthy to remind that $k$ is the subscript for iteration index instead of that for vector component. In this paper no confusion results because which meaning is intended is always clear from the context in which it appears.

Our first task here is to prove the convergence of Gauss-Newton Algorithm. From now on, our vim  focuses on the second optimization, so the first optimization is always assumed to be feasible and solvable. Mathematically, the objective function is assumed to have fairly good analytical properties, such as continuity, differentiability, Lipschitz continuity over certain neighborhood, the positive definite of matrix, and so forth. Some subsidary mathematical materials are compiled in the appendix. Nevertheless, two lemmas that are needed as the prerequisites of convergence proof, are presented below.

\begin{lemma} Let $A$ and $B$ be $n\times n$ matrices, let $A$ be
nonsingular and $\| A^{-1} \| \leq \alpha$, let $\| B-A \| \leq \beta$, and let $\alpha \beta \le 1$, then $B$ is nonsingular and
\beq\label{eq:bivs} \| B^{-1} \| \leq \frac{\alpha}{1-\alpha \beta}.\eeq
\label{lemma1}\end{lemma}
\begin{proof}  The first step is to establish equality $ B^{-1}=\sum\limits^{\infty}_{k=0} (I-A^{-1}B)^{k}A^{-1} $. \\
Notice $(I-A^{-1}B)^{k} = [A^{-1}(A-B)]^{k} $ ,
$\| A^{-1}(A-B) \| \leq \alpha \beta \le 1 $ ,\\
whence $ \sum\limits^{\infty}_{k=0} (I-A^{-1}B)^{k}= [I-(I-A^{-1}B)]^{-1}= B^{-1} A $. Multiple both sides with $A^{-1} $, get the needed result.
The second step is to prove ${\displaystyle \|\sum\limits^{\infty}_{k=0} (I-A^{-1}B)^{k}A^{-1}\|\leq \frac{\alpha}{1-\alpha \beta} }$.\\
$$\|\sum^{\infty}_{k=0} (I-A^{-1}B)^{k}A^{-1}\|\leq \|A^{-1}\| \| \sum^{\infty}_{k=0} (I-A^{-1}B)^{k}\| \leq \frac{\|A^{-1}\|}{1-\| I-A^{-1}B\|} \leq
\frac{\alpha}{1-\alpha \beta} .$$\\
\end{proof}

\begin{lemma} Let $F:R^n \to R^m$ be a continuous and differentiable function over a open convex set $D \subset R^n$, the derivative of $F$, i.e. $F^{\prime}$, is Lipschitz continuous over a neighborhood of any $\theta \in D$, then for any $\theta + \delta \in D$, we have
\beq \|F(\theta + \delta)-F(\theta)-F^{\prime}(\theta)\delta\|
\leq \frac{\gamma}{2} \|\delta\|^2, \eeq
where $\gamma$ is a Lipschitz constant.
\label{lemma2}\end{lemma}
\begin{proof}
\[\begin{array}{rcl} F(\theta + \delta)-F(\theta)-F^{\prime}(\theta)\delta &=&
 {\displaystyle \int^{1}_{0} F^{\prime}(\theta+t\delta )\delta dt - F^{\prime}(\theta)\delta} \\
  &=& {\displaystyle \int^{1}_{0} [ F^{\prime}(\theta+t\delta )- F^{\prime}(\theta)] \delta dt },
\end{array}\]
therefore
\[\begin{array}{rcl} \|F(\theta + \delta)-F(\theta)-F^{\prime}(\theta)\delta\| &\leq&
 {\displaystyle \int^{1}_{0} \| F^{\prime}(\theta+t\delta )- F^{\prime}(\theta)\|\|\delta\| dt }\\
 &\leq&  {\displaystyle \int^{1}_{0} \gamma \| t\delta \| \|\delta\| dt } \\
 &=& {\displaystyle \gamma \| \delta \|^2 \int^{1}_{0} t dt =\frac{\gamma}{2} \| \delta \|^2 }.
\end{array}\]
\end{proof}

\begin{theorem}[Convergence of Gauss-Newton Algorithm] If $f:R^n \to R$ has a continuous second partial derivative over a open convex set $D \subset R^n$, $J(\theta)$ is Lipschitz continuous over $D$, i.e. $\|J(\phi)-J(\theta)\|\leq \gamma \|\phi-\theta\|$, $\forall \phi, \theta \in D$; and $\|J(\theta)\|\le \alpha$, $\forall \theta \in D$. If there exists a critical point $\thetast$ such that $J(\thetast)^T r(\thetast)=0$, then the sequence $\{ \theta_{k}\}$ generated by Gauss-Newton Algorithm converges to $\thetast$.
\label{theorem1}\end{theorem}

\begin{proof}  Herein simplified symbols are introduced for the following proof,
$$\theta_k=(\theta_1^k,\theta_2^k,\cdots,\theta_n^k)^T, \thetast=(\thetast_1,\thetast_2,\cdots,\thetast_n)^T;  $$
$$J=J(\theta),  J_k=J(\theta_k), \jst=J(\thetast), r_k=r(\theta_k), \rst=r(\thetast).   $$
By $\jst^T \rst=0$ and $J$ satisfies Lipschitz continuous condition, we have
\beq \|(J-\jst)^T\rst\|\leq \eta \|\theta-\thetast\|. \label{eq:jjstineq}\eeq
Let $\lambda$ be the smallest eigenvalue of $\jst^T \jst$, then there exists $\epsilon_1$ such that $\eta \le \lambda$ when $\theta \in N(\thetast,\epsilon_1 )$. Due to Lipschitz continuity of $J$, there exists $\epsilon_2$, when $\theta \in N(\thetast,\epsilon_2)$ such that
$$\|J^T J-\jst^T \jst\| < \kappa \leq \lambda - \eta .$$
Notice ${\displaystyle \|(\jst^T \jst)^{-1}\|\leq \frac{1}{\lambda} },$
by \lmref{lemma1},
\beq  \|(J^T J )^{-1}\| \leq \frac{c}{\lambda} \mbox{~~with~~}
c \in (1,\frac{\lambda}{\eta}).
\label{eq:jjins}\eeq
Take ${\displaystyle \epsilon = \min \{\epsilon_1,\epsilon_2, \frac{\lambda-c\eta}{c\alpha\gamma}\}} $, suppose after $k$ step iterations, ${\displaystyle \|(J_k^T J_k)^{-1}\| \leq \frac{c}{\lambda} },$ then the $(k+1)$-th step $\theta_{k+1}$ has definition, and
\beq \begin{array}{rcl}
\theta_{k+1} -\thetast &=& \theta_{k} -\thetast-(J_k^T J_k)^{-1} J_k^T r_k  \\
   &=&-(J_k^T J_k)^{-1} [J_k^T r_k +J_k^T J_k (\thetast-\theta_{k})]\\
   &=&-(J_k^T J_k)^{-1} [J_k^T \rst- J_k^T(\rst- r_k - J_k (\thetast-\theta_{k}) )].
\end{array} \label{eq:stepkk1}\eeq
Hence by \lmref{lemma2},
\beq  \|\rst- r_k - J_k (\thetast-\theta_{k})\| \leq \frac{\gamma}{2} \|\theta_{k}-\thetast\|^2.
\label{eq:rrkineq}\eeq
By hypothesis $\jst^T \rst=0$ and relation \eref{eq:jjstineq},
\beq \|J_k^T \rst \|=\|(J_k-\jst)^T\rst\|\leq \eta \|\theta_k-\thetast\|~. \label{eq:jkjstineq}\eeq
Synthesizing relations \eref{eq:jjins}, \eref{eq:rrkineq}, \eref{eq:jkjstineq}, and hypothesis  $\|J_k\|\le \alpha$, in the light of formula \eref{eq:stepkk1}, it is readily to get
\beq \begin{array}{rcl}
\|\theta_{k+1} -\thetast\| &\leq &\|(J_k^T J_k)^{-1}\| (\|J_k^T \rst\| + \|J_k \| \|\rst- r_k - J_k (\thetast-\theta_{k})\|)  \\
   &\leq& {\displaystyle \frac{c}{\lambda} (\eta\|\theta_{k}-\thetast\| + \frac{\alpha\gamma}{2} \|\theta_{k}-\thetast\|^2) } \\
   &\leq& {\displaystyle (\frac{c\eta}{\lambda}+ \frac{\lambda-c\eta}{2\lambda} ) \|\theta_{k}-\thetast\| < \|\theta_{k}-\thetast\|} ,
\end{array} \label{eq:stepkk1norm}\eeq
that is
$$\|\theta_{k+1} -\thetast\| = \rho \|\theta_{k}-\thetast\| \mbox{~~with~~} \rho <1 .$$
In the preceding proof, if let $k=0$, we immediately obtain the proof for the first step of induction conjecture, therefore, according to principle of mathematical induction the above proof is right for any $k$. Whence
$$\|\theta_{k} -\thetast\| = \rho^{k+1} \|\theta_{0}-\thetast\|,$$
when $k \to \infty$, $\|\theta_{k} -\thetast\|  \to 0$, that is
$$\lim\limits_{k\to \infty} \theta_k = \thetast,$$
which indicates that the sequence $\{ \theta_{k}\}$ generated by Gauss-Newton Algorithm converges to $\thetast$.
\end{proof}

\begin{lemma}[Affine Invariance of Step] For Gauss-Newton algorithm, iteration step is independent on affine transformation.
\label{lemma3}\end{lemma}

\begin{proof}  Let $U$ be an $n\times n$ nonsingular matrix, define $\bar{f}(\phi)=f(U\phi)$, then we have
\beq \begin{array}{rcl}
\nabla \bar{f}(\phi)&=& U^T \nabla f(\theta)~, \\
\nabla^2 \bar{f}(\phi)&=& U^T \nabla^2 f(\theta) U~, \\
\theta &=& U \phi~.
\end{array} \eeq
Then the step of Gauss-Newton algorithm can be expressed as
\beq \begin{array}{rcl}
\delta \phi &=& - [U^T \nabla^2 f(\theta) U]^{-1} [U^T  \nabla f(\theta)]~ , \\
            &=& - U^{-1}[\nabla^2 f(\theta) \cdot \nabla f(\theta) ]~ , \\
            &=& U^{-1} \delta \theta~,
\end{array} \eeq
that is the step is independent on affine transformation. Especially for the least square model, the corresponding transformation has the following form
\beq \begin{array}{rcl}
\nabla [J^T r] (\phi)&=& U^T [J^T r] (\theta)~, \\
\nabla^2 [J^T J] (\phi)&=& U^T [J^T J] (\theta) U~, \\
\phi &=& U^{-1} \theta~,
\end{array} \eeq
by virtue of which all requirements in \thref{theorem1} are satisfied for the transformed variables and the proof of \thref{theorem1} remains valid, so
\beq \lim\limits_{k\to \infty} \phi_k = \phist, \mbox{~~with~~~} \phist = U^{-1} \thetast~ . \eeq
\end{proof}

\noindent {\bf Remark} In the light of \thref{theorem1}, Gauss-Newton algorithm is merely related to the reduced Hesse matrix $H(\theta)$, the corresponding error estimation is also merely related to $H(\theta)$, therefore from the conventional definition of covariance matrix (denoted by $V$), it reads
\beq  V^{-1}(\theta,x) \equiv H(\theta,x)~. \label{defcvmtx}\eeq
The index $x$ is recovered henceforth. By formulas \eref{eq:ridef}, \eref{jdef}, and \eref{defrdchesmtx}, the Hessian element is expressed as
\beq  H_{ij}(\theta,x)=L \epsilon \cdot \sum\limits^{m}_{l=1} \frac{x_l}{\bar{\sigma}_l}
\frac{\partial \sigma_l}{\partial \theta_i} \frac{\partial \sigma_l}{\partial \theta_j}~. \label{defhij}\eeq
Define $H(\theta,x)=L \epsilon \tilde{H}(\theta,x)$, recast $\tilde{H}(\theta,x)$ as
\beq \tilde{H}(\theta,x)=A Y A^T, \label{defhtilde}\eeq
where
\beq A_{ij}= \frac{\partial \sigma_j}{\partial \theta_i}~,~
     A^T_{ij}= \frac{\partial \sigma_i}{\partial \theta_j}~,~
     Y_{ij}=\frac{x_i}{\bar{\sigma}_i} \cdot \delta_{ij}~. \label{expnayat}\eeq
Since there is only a constant ($L \epsilon$) difference between $H$ and $\tilde{H}$, the latter is often adopted in the following discussions. However, the conclusions are simultaneously applicable for the former.

By virtue of \lmref{lemma3}, affine transformation will not change optimization result, so it is useful to investigate the property of $\tilde{H}$ under special affine coordinate, which will disclose some peculiar features of optimization.

\begin{theorem}[Independence of Optimal Parameters] For each parameter, there exists an affine transformation under which a parameter is relatively independent of the others.
\label{theorem2}\end{theorem}

\begin{proof}
Since $\tilde{H}$ is a symmetric positive definite matrix, according to Cholesky decomposition theorem, there exists an affine transformation, such that
\beq \tilde{H}=\hat{L} D \hat{L}^T~, \eeq
where $D$ is a diagonal matrix, $\hat{L}$ is a unit lower triangle matrix, the inverse of which, $\Gamma \equiv \hat{L}^{-1}$, is also a unit lower triangle matrix. Whence
\beq D= \Gamma \tilde{H} \Gamma^T~. \eeq
By \lmref{lemma3}, this affine transformation is equivalent to introduce a new variable $\phi$, with the relation
\beq \phi = (\Gamma^T)^{-1} \theta~. \eeq
Since $D$ is diagonal, this indicates that every element of $\phi$, {\em i.e.} $\phi_i, (i=1,2, \cdots, n)$, is independent of the other elements. Note that $\Gamma^T$ is a unit upper triangle matrix, its inverse is also a unit upper triangle matrix that has the form
\[
(\Gamma^T)^{-1}=\left(
\begin{array}{cccc}
1 & t_{21} &\cdots & t_{n1} \\
  & 1      &\cdots & t_{n1} \\
  &        &\ddots &\vdots \\
  &        &       & 1
\end{array}
\right),
\]
by which we have
$$\phi_n=\theta_n~.$$
Since the order of parameter, that is the index of parameter, has not absolutely assigned meaning, the above statements indicate that for any parameter $\theta_i~(i=1,2, \cdots, n)$ there exists an affine transformation such that $\theta_i$ is independent of the others.
\end{proof}

\noindent {\bf Remark} Since we can always have $\phi_i=\theta_i$, this theorem indicates that from the point view of optimization process, we can directly utilize original parameters instead of actually performing an affine transformation. In addition, the ``independence conjecture'' adopted in Sect.~\ref{sect_mltpzn} now becomes the bona fide legitimate theorem.

\begin{theorem}[Uniqueness of Minimum] As far as all taken scan points are concerned, there is only one point which leads to the smallest error for certain parameter.
\label{theorem3}\end{theorem}

\begin{proof}
Firstly introduce an {\bf auxiliary function $g$} with definition
\beq  g^{l}_{ij}(\theta)=  \frac{1}{\bar{\sigma}_l}
\frac{\partial \sigma_l}{\partial \theta_i} \frac{\partial \sigma_l}{\partial \theta_j}~. \label{defgfkn}\eeq
Here by \thref{theorem2}, without loss of generality, only consider parameter $\theta_1$ and keep the other parameters invariant, then $g$ is simplified as
\beq g^{l}_{11}(\theta_1)=  \frac{1}{\bar{\sigma}_l}
\left(\frac{d \sigma_l}{d \theta_1} \right)^2~,  \eeq
and $\tilde{H}$ becomes
\beq \tilde{H}(\theta_1,x) = \sum\limits^{m}_{l=1}  x_{l} g^{l}_{11}(\theta_1)~. \eeq
Now denote $g_{min}=\min\{g^{l}_{11}, l=1,2, \cdots, m\}$ and $g_{max}=\max\{g^{l}_{11}, l=1,2, \cdots, m\}$, notice $\sum\limits^{m}_{l=1} x_{l} =1$, it is easy to see
%\beq g_{min} = \left(\sum\limits^{m}_{l=1} x_{l}\right) g_{min} \leq \sum\limits^{m}_{l=1} x_{l} %g^{l}_{11} \leq \left(\sum\limits^{m}_{l=1} x_{l}\right) g_{max} = g_{max}. \eeq
\beq g_{min} \leq \left(\tilde{H}= \sum\limits^{m}_{l=1} x_{l} g^{l}_{11} \right) \leq g_{max}~. \eeq
Notice $V \propto \tilde{H}^{-1}$, from above relation, when $\tilde{H}$ reaches the maximizer
$V$ reaches the minimizer. Therefore, if $g^{l}_{11}(\theta_1)$ exists only one maximum value on all scan points (for $l=1,2, \cdots, m$), then when all luminosity is congregated at the point where $g=g_{max}$, the minimum error can be obtained for parameter $\theta_1$.
\end{proof}

\noindent {\bf Remark} Generally speaking, for scan experiments, the cross section $\sigma_l$ is a function of center-of-mass energy $E$ (denoted as $E_{cm}$ in Sect.~\ref{sxn:samplingmed}), each index $l$ actually corresponds to an energy $E_l$, that is $\sigma_l=\sigma_{E_l}$. However, the $E_l$ is a continuous variable and can be denoted directly as $E$, so the cross section reads $\sigma (\theta;E)$. Correspondingly $g$ is a function of $E$, that is
\beq g(\theta ; E)=g^{E}_{11}(\theta)= \frac{1}{\sigma (\theta;E)} \cdot
\left[ \frac{d \sigma (\theta;E)}{d \theta} \right]^2~,  \eeq
where $\bar{\sigma}_l \approx \sigma_{E_l}=\sigma(\theta;E)$ is adopted. So the extremum of $g(\theta;E)$ can be acquired by setting $dg(\theta;E)/dE=0$.

\begin{theorem}[Number Consistency between parameters and scan points] Solely $n$ points are needed for an optimal scan scheme in which $n$ parameters are to be determined. The energy value of each scan point corresponds to the position where an auxiliary function of certain parameter reaches its maximizer.
\label{theorem4}\end{theorem}

\begin{proof}
Firstly, we prove the following fact. Considering the summation
\beq T_m= \sum\limits^{m}_{t=1} z_{j_t} g^{j_t}_{i}~,  \eeq
%where $g^{j}_{i}$ is shorthand for  $$ g^{j}_{i} = g^{E_j}_{ii} = $$
where $i$ denotes a certain parameter and $j_t$ an energy point. Without loss of generality, assume that $j_1 \le j_2 \le \cdots \le j_m$, and denote $g_{min}=\min\{g^{j_t}_{i}, t=1,2, \cdots, m\}$ and $g_{max}=\max\{g^{j_t}_{i}, t=1,2, \cdots, m\}$, then
\beq z g_{min} \le T_m\le z g_{max}~,  \eeq
with $\sum\limits^{m}_{t=1} z_{j_t} =z$. Notice the continuous dependence of $g^{j_t}_{i}$ on $j_t$ (as explained in the remark of \thref{theorem3}, $j_t=E_{j_t}$ and $E_{j_t}$ is a continuous variable, so $g^{j_t}_{i}$ is a continuous function of $E_{j_t}$), therefore it is always possible to find a point $j^{\prime}$ with $j_1 \leq j^{\prime} \leq j_m$ so that
\beq z g^{j^{\prime}} = T_m~.  \eeq
As shown in \thref{theorem3}, $g^{j_t}_{i}$ is related to uncertainty of parameter. Therefore, above fact indicates that the uncertainty effects due to $m$ points can be realized solely at one point. Next considering $\tilde{H}$.

Since $\tilde{H}$ is a symmetric matrix, there exists an orthogonal matrix $U$ such that
\beq U^T \tilde{H} U = D~, \eeq
where $D$ is a diagonal matrix. It is assumed that for $n$ parameters in $D$-space, $n$ optimal points have been figured out according to \thref{theorem3}. Although these optimal positions of points may be distinct from those of $\tilde{H}$, by virtue of \lmref{lemma3}, two kinds of positions are of one-to-one correspondence. We investigate any one of diagonal elements of $D$, say $D_{i}$, which has the form
\beq D_{i}= \sum\limits^{n}_{l=1} x_{l} g^{l}_{i}+ y  g^{j}_{i}~, \eeq
where $\sum\limits^{n}_{l=1} x_{l}+y =1$, and $j\neq l, (l=1,2, \cdots, n)$; $l$ denotes optimal point, so $y$ indicates the luminosity allocated outside the optimal points; $g^{l}_{i}$ is defined as
\beq g^{l}_{i}=  \frac{1}{\bar{\sigma}_l}
\left(\frac{\partial \sigma_l}{\partial \phi_i} \right)^2. \eeq
Then we try to find out an allocation $\sum\limits^{n}_{l=1} y_l=y$, such that for each parameter
\beq \sum\limits^{n}_{l=1} y_l g^{l}_{i} \geq y g^{j}_{i}~; \eeq
or for all diagonal elements of $D$
\beq G Y \geq Y_g~, \label{eq:gmatrix} \eeq
with
\[
G =\left(
\begin{array}{cccc}
 g_1^1  & g_1^2  &\cdots & g_1^n  \\
 g_2^1  & g_2^2  &\cdots & g_2^n  \\
 \vdots & \vdots &\ddots &\vdots \\
 g_n^1  & g_n^2  &\cdots & g_n^n
\end{array}
\right)~,
\]
$Y= (y_1,y_2, \cdots, y_n)^T~,$ and $Y_g= (y g^{j}_{1}, y g^{j}_{2}, \cdots, y g^{j}_{n})^T~.$

If Eq.~\eref{eq:gmatrix} is satisfied, this indicates that the allocation of luminosity at optimal points leads to increase of $D_{i}$. Notice that $V_{i} \propto D_{i}^{-1}$, therefore, this allocation of luminosity leads to decrease of error. Moreover, note
\beq D^{-1} = U^{-1} \tilde{H}^{-1} (U^T)^{-1} , \eeq
since orthogonal transformation keeps the trace of matrix invariant, while the trace of $\tilde{H}^{-1}$ is proportional to the sum of squared errors, the previous statements indicate that the luminosity allocated at $n$ optimal points can guarantee the optimality of results for determining $n$ parameters. Next we prove Eq.~\eref{eq:gmatrix}.

Firstly, introduce new variables
$$\rho_i^l = \frac{g_i^l}{g_i^j}~,~~
z_i = \frac{y_i}{y}~,~~
I=(\underbrace{1,1,\cdots, 1}_{n})^T~.$$
Since for the parameter $i$, the optimal point is at the position $i$, $g_i^i > g_i^k ~(k=l ~\text{or} ~ j )$ and $\rho_i^i >  1, \rho_i^i >  \rho_i^l > 0~.$ The new matrix $P(t)$ is defined as
\[
P(t) =\left(
\begin{array}{cccc}
 \rho_1^1   & \rho_1^2 t &\cdots & \rho_1^n t \\
 \rho_2^1 t & \rho_2^2   &\cdots & \rho_2^n t \\
 \vdots & \vdots &\ddots &\vdots \\
 \rho_n^1 t & \rho_n^2 t &\cdots & \rho_n^n
\end{array}
\right)~,
\]
or
$$[P(t)]_{pq}=\rho_p^q [\delta_{pq}+(1-\delta_{pq})t]~. $$
Here $t\in [0,1]$. Then the equivalent equation of Eq.~\eref{eq:gmatrix} is
\beq P(1) Z \geq I~, \label{eq:pmatrix} \eeq
with $Z= (z_1,z_2, \cdots, z_n)^T~.$ Considering the case when $t=0$ and $P(0) Z = I~$, it immediately gets $ \rho_i^i z_i =1 $ or  $  z_i =1/\rho_i^i < 1 $ (since $\rho_i^i > 1$). For most of cases~\cite{wangyk2007,wangbq2012} $g_i^i \gg 1$, and $g_i^j$ should be far from $ g_i^i$, which means $\rho_i^i \gg 1$. Under such a case, we consider the condition
\beq \sum\limits^{n}_{i=1} \frac{1}{\rho_i^i } \leq 1~, \label{eq:condition} \eeq
which is equivalent to $\sum\limits^{n}_{i=1} z_i \leq 1 $. For $\sum\limits^{n}_{i=1} z_i < 1 $ case, it always can increase some $z_i$ to make $\sum\limits^{n}_{i=1} z_i = 1 $. Therefore, it always has $P(0) Z \geq I$ for $\sum\limits^{n}_{i=1} z_i = 1 $. Notice $P(t)Z$ is an increase function on $t$, so
$$P(1) Z \geq P(0) Z \geq I~.$$
This finish the required proof.
\end{proof}

\noindent {\bf Remark} Firstly, the application of Cramer's theorem immediately yields the solution of equation $G Y = Y_g$. But such a solution can not guarantee $Y \ge 0$. Secondly, conclusions on positive linear system can guarantee a positive solution~\cite{Kaykobad1985a,Kaykobad1985b}. However, the constraint $\sum\limits^{n}_{l=1} y_l=y$ makes these conclusions not feasible for the present problem. Thirdly, the condition of Eq.~\eref{eq:condition} greatly simplifies the proof of \thref{theorem4}. Nevertheless, as from above proof, \thref{theorem4} can be actually applied for some cases when $\sum\limits^{n}_{i=1} 1/\rho_i^i > 1$. The applicable degree of inequality depends on the variation character of $\rho_i^k ~(k=l ~\text{or} ~ j)$. Anyway, further mathematical study of this problem is beyond the interest of this paper.

\begin{theorem}[Luminosity Distribution Principle] For multi-parameter scan scheme, the luminosity allocation among points is relevant to the relative importance between parameters, cross section and its derivative to parameter at optimal point.
\label{theorem5}\end{theorem}

\begin{proof}  We begin with definition \eref{defcvmtx}, {\em i.e.} $V^{-1}= H$, by formulas \eref{defhij} and \eref{defhtilde}, $V^{-1}$ can be recast as
\beq V^{-1} =A Z A^T~, \label{defvinvs}\eeq
where $A$ is given in Eq. \eref{expnayat},
\beq Z_{ij}=\frac{L \epsilon x_i}{\bar{\sigma}_i} \cdot \delta_{ij}~. \label{expnz}\eeq
We know $\bar{\sigma}_i$ is observed cross section. After the first optimization, we obtained optimal parameters and have relation $\sigmast_i\equiv \sigma_i (\thetast) \approx \bar{\sigma}_i$. In the following theoretical analysis, we replace $\bar{\sigma}_i$ with $\sigmast_i$, then
\beq V =(A^T)^{-1} Z^{-1} A^{-1}~, \label{defv}\eeq
where
\beq Z^{-1}_{ij}=\frac{\sigmast_i} {L \epsilon x_i}\cdot \delta_{ij}~. \label{expnzinvs}\eeq
The element of $A^{-1}$ is represented by $\alpha$, {\em i.e.} $\alpha_{ij} \equiv A^{-1}_{ij}$. The diagonal elements of $V$ is the squared error for certain parameter, for the $i$-th parameter, the squared error reads
\beq v_{ii}= \frac{1}{L \epsilon} \cdot \sum\limits^{n}_{l=1} \frac{\alpha^2_{li} \sigmast_l}{x_l}~. \label{expnerrorsq}\eeq
For each $v_{ii}$ introduce a weight factor $w_i$ to represent the relative importance of parameter. Notice constraint $\sum\limits^{n}_{l=1}x_l=1$, introduce a Lagrange multiplier
$\lambda$, and construct a Lagrange function as follows
\beq \begin{array}{rcl}
F(\theta,x;\lambda)&=& \sum\limits^{n}_{i=1} w_{i} v_{ii} + \lambda \left(\sum\limits^{n}_{l=1}x_l -1 \right)\\
&=&{\displaystyle \sum\limits^{n}_{i=1} w_{i}
\left(\frac{1}{L \epsilon} \cdot \sum\limits^{n}_{l=1} \frac{\alpha^2_{li} \sigmast_l}{x_l}\right)
+ \lambda \left(\sum\limits^{n}_{l=1}x_l -1\right) }~.
\end{array} \eeq
The first derivative of $F$ leads to
\beq \frac{\partial F}{\partial x_p} =- \frac{1}{L \epsilon} \cdot \frac{\sigmast_p}{x_p^2} \cdot
\left(\sum\limits^{n}_{i=1} w_{i} \alpha^2_{pi} \right)+ \lambda~ ;
\label{expnfirstf}\eeq
the second derivative of $F$ leads to
\beq \frac{\partial^2 F}{\partial x_q \partial x_p} =\frac{1}{L \epsilon} \cdot \frac{2\sigmast_p}{x_p^3} \cdot
\left(\sum\limits^{n}_{i=1} w_{i} \alpha^2_{pi} \right) \delta_{pq}~.
\label{expnsecondf}\eeq
Since the Hessian of $F$ only has diagonal elements and all of them are positive, the extremum determined by the first derivative are all minimizers. By setting Eq.\eref{expnfirstf} to zero, it is readily to get
\beq x_p^2 =\frac{\sigmast_p}{L \epsilon \lambda}  \cdot
\left(\sum\limits^{n}_{i=1} w_{i} \alpha^2_{pi} \right),
\label{expnxpsq}\eeq
on the strength of which the luminosity allocation among different optimal points can be obtained.
\end{proof}

\noindent {\bf Remark} The weight factor $w_i$ is determined usually according to antecedent experience. Take $\taums$ scan as an example, as to parameter 1 ($\taums$), 2 ($\epsilon$), and 3 ($\sigbg$), the corresponding weight is set to be $w_1=1$, $w_2=w_3=0$, which indicates that only the uncertainty of $\taums$ is cared about. Under such a condition, \thref{theorem5} yields the following
relation
\begin{equation}
\label{eq:taumsptratio}
x_1:x_2:x_3 = (\alpha_{11}\sqrt{\sigmast_1}) :
(\alpha_{21}\sqrt{\sigmast_2}) : (\alpha_{31}\sqrt{\sigmast_3})~.
\end{equation}
The cross sections and their derivatives are calculated respectively at energy points~\cite{wangyk2009} $E_1$ = 3.5538 GeV, $E_2$ = 3.595 GeV, and $E_3$ = 3.50 GeV. The optimal fraction of luminosity at these points are $x_1$ = 70.0\%, $x_2$ = 21.8\%, and x$_3$ = 8.2\%. Comparing with the results by the sampling technique (refer to Ref.~\cite{wangyk2009} or Sect.~\ref{sxn:samplingmed}) $x_1$ = 67.5\%, $x_2$ = 22.5\%, and $x_3$ = 10.0\%, two sets of results are consistent with each other fairly well. In the calculation, some relevant values are kept the same as those used in Ref.~\cite{wangyk2009}: $m_{\tau}$ = 1.77699 GeV, $B_{e\mu}$ = 0.06194, $\varepsilon$ = 14.2\%, and $\sigma_{BG}$ =0.024 pb.

\section{Discussion}\label{sxn:discussion}
\subsection{Equivalence between likelihood and chisquare fits}\label{sxteqlkandsq}
We start from likelihood estimator (refer to Eq.~\eref{lklihd}),
\begin{equation}
\mathrm{LF} = \prod_{l=1}^{n} \frac{\mu_l^{N_l} e^{-\mu_l}} {N_l!}~,
\label{lklihd2}
\end{equation}
where $N_l$ is the number of observed events at $l$-th scan point,
and $N_i$ is assumed following a Poisson distribution with expectation $\mu_l$.
To find the maximum of likelihood function equals to find the minimum of
function $f$ defined as
\begin{equation}
f = -\ln\mathrm{LF} = - \sum_{l=1}^{n} \ln\left(\frac{\mu_l^{N_l}
    e^{-\mu_l}} {N_l!} \right).
\label{eq:fdefasmlf}
\end{equation}
With definition $$f_l \equiv \frac{\mu_l^{N_l}  e^{-\mu_l}} {N_l!}~, $$
the second order derivative of function $f$ reads
\begin{equation}
\frac{\partial^2 f} {\partial \theta_i \partial \theta_j} =
\sum_{l=1}^{n} \left[ \left( \frac{1}{f_l} \frac{\partial f_l} {\partial
    \theta_i} \right) \left( \frac{1}{f_l} \frac{\partial f_l}
  {\partial \theta_j} \right) - \frac{1}{f_l} \frac{\partial^2
  f_l}{\partial \theta_i \partial \theta_j} \right]~.
\label{eq:skddoff}\end{equation}
After a little algebra, this equation is reduced to
\beq
\frac{\partial^2 f} {\partial \theta_i \partial \theta_j} = \sum_{l=1}^{n} \left\{ \frac{N_l}{\mu_l^2} \frac{\partial\mu_l}{\partial\theta_i} \frac{\partial \mu_l}{\partial \theta_j} + \left(\frac{N_l}{\mu_l}-1\right) \frac{\partial^2 \mu_l}{\partial\theta_i \partial\theta_j}   \right\}~.
\label{eq:sdfexpa}\eeq
Notice that for Poisson distribution the expectation of $(N_l-\mu_l)^2$ is $\mu_l$, for large $N_l$ we take approximation $(N_l-\mu_l)^2 \approx \mu_l$, so it is easy to see
$$\left(\frac{N_l}{\mu_l}-1\right) \approx \frac{1}{\sqrt{\mu_l}}, \mbox{~~ for large~~} N_l~. $$
In addition, utilizing relation $\mu_l = L_l \sigma_l$ ($\sigma$ denotes the theoretical cross section), $N_l = L_l \bar{\sigma}_l$ ($\bar{\sigma}$ denotes the observed cross section), and after the first optimization, taking approximation $\sigma_l \approx \bar{\sigma}_l$, Eq.~\eref{eq:sdfexpa} can be recast as
\beq
\frac{\partial^2 f} {\partial \theta_i \partial \theta_j} = \sum_{l=1}^{n} \left\{ \frac{N_l}{\bar{\sigma}_l^2} \frac{\partial\mu_l}{\partial\theta_i} \frac{\partial \mu_l}{\partial \theta_j} - \sqrt{\frac{N_l}{\bar{\sigma}_l^2}} \frac{\partial^2 \mu_l}{\partial\theta_i \partial\theta_j}   \right\}.
\label{eq:sdfexpb}\eeq
Since $\sigma$ is the physics quantity which keeps invariant for a definite process, therefore when $N_l$ is large enough, it always satisfies
$$\frac{N_l}{\bar{\sigma}_l^2} \gg \sqrt{\frac{N_l}{\bar{\sigma}_l^2}}. $$
This indicates that comparing with the first term of Eq.~\eref{eq:sdfexpb}, the second term can be neglected, which leads to
\begin{equation}
\frac{\partial^2 f} {\partial \theta_i \partial \theta_j} \approx L \epsilon
\sum_{l=1}^{n} \frac{x_l} {\bar{\sigma}_l} \left( \frac{\partial \sigma_l}
  {\partial \theta_i} \right) \left( \frac{\partial \sigma_l}
  {\partial \theta_j} \right).
\label{eq:fijexprn}\end{equation}
Here relation $N_l=\epsilon L x_l \bar{\sigma}_l$ is adopted. Comparing with Eq.~\eref{defhij}, ${\displaystyle \frac{\partial^2 f} {\partial \theta_i \partial \theta_j}}$ is just the element of Hesse matrix ($H_{ij}$). Moreover, both likelihood and chisquare estimators have the same form of gradient relevant to $\sigma_l$, therefore Gauss-Newton Algorithm can be executed for likelihood estimator exactly the same way as that for chisquare. On this extent, it is reasonable to claim that the first optimization processes for both likelihood and chisquare fits are equivalent.

As a matter of fact, we could view the equivalence between likelihood and chisquare fit from another viewpoint. Notice for large $N$, Poisson distribution approximates Gauss distribution, $i.e.$
\beq
\frac{\mu^{N} e^{-\mu}} {N!} \xrightarrow{N\to \infty} \frac{1}{\sqrt{2\pi E_N}}
e^{-\frac{(N-\mu)^2}{2 E_N}},
\label{psnandgauss}\eeq
where $E_N$ is expectation of $N$, if we take $E_N \approx N$, then
\beq
\mathrm{LF} = \prod_{l=1}^{n} \frac{\mu_l^{N_l} e^{-\mu_l}} {N_l!}
 \xrightarrow{N\to \infty}
\mathrm{LF} = \prod_{l=1}^{n} \frac{1}{\sqrt{2\pi N_l}}
e^{-\frac{(N_l-\mu_l)^2}{2 N_l}}.
\label{lklihdforpandg}\eeq
Whence
\beq
f = -\ln\mathrm{LF} = \frac{1}{2} \sum_{l=1}^{n} \frac{(N_l-\mu_l)^2}{N_l}+
    \frac{1}{2}  \sum_{l=1}^{n} \ln(2\pi N_l).
\label{fforpandga}
\eeq
In optimization process, the second term as a constant can be neglected, so
$f$ becomes
\beq
f = -\ln\mathrm{LF} = \frac{1}{2} \sum_{l=1}^{n} \frac{(N_l-\mu_l)^2}{N_l},
\label{fforpandgb}
\eeq
so except for a factor $1/2$, this is the chisquare form in Eq.~\eref{eq:chisq}, which is adopted from the very beginning of study.

\subsection{Effect due to systematic uncertainty}\label{sxtsyserror}
%\subsection{Effect due to energy uncertainty}
In the light of study of $\taums$ scan, the uncertainty due to energy calibration dominates over the others~\cite{Mo:2007npb,Fu2008}. Some special techniques have been adopted to decrease such an uncertainty. For example, Compton backscattering technique is utilizing to establish beam energy measurement system at KEDR~\cite{bemskder} and BES~\cite{bemsepc,Mo2008}, to increase the accuracy of beam energy at the level of $10^{-4}$ or better.

There is a concise way to taking into account of such a kind of uncertainty. We begin with chisquare formula~\eref{eq:chisq}
\begin{equation}
\chi^2 = \sum_{i=1}^{n} \left( \frac{\bar{N}_i - N_i} {\bar{\Delta}_i} \right)^2~,
\label{eq:chisqb}\end{equation}
where $\bar{N}$ denotes the number of observed events, ${N}$ the number of theoretical estimated events, $\bar{\Delta}$ the error of $\bar{N}$, then according to Refs.~\cite{Barker1974,Orear1982,Roe2001}, the effect of uncertainty due to energy $E$ will taken into account by a new chi-square form
\begin{equation}
\chi^2_E = \sum_{i=1}^{n} \left( \frac{\bar{N}_i - N_i} {\tilde{\Delta}_i} \right)^2~,
\label{eq:chisqeg}\end{equation}
where
\beq
\tilde{\Delta}_i^2 = \bar{\Delta}_i^2 +\left[ \left.\frac{d N}{d E}\right|_{E=E_i} \cdot \Delta_{E_i} \right]^2~,
\label{eq:tlddelta}\eeq

Notice $\bar{\Delta}_i^2=\bar{N}_i$, and $N_i(\bar{N}_i)=L_i \sigma_i (\bar{\sigma}_i)$, with following definitions
$$\hat{\sigma}_{E_i} \equiv E_i \cdot \left.\frac{d \sigma}{d E}\right|_{E=E_i}~,~~
\delta_{E_i} \equiv \frac{\Delta_{E_i}}{E_i}~,  $$
Eq.~\eref{eq:tlddelta} is recast as
\beq
\begin{array}{rcl}
\tilde{\Delta}_i^2 &=&\bar{N}_i + L_i^2 (\hat{\sigma}_{E_i})^2 \delta_{E_i}^2 \\
                   &=&L_i \tilde{\sigma}_i~,
\end{array}
\label{eq:tlddeltab}\eeq
with definition
\beq
\tilde{\sigma}_i \equiv \bar{\sigma}_i+ L_i (\hat{\sigma}_{E_i})^2 \delta_{E_i}^2~.
\label{eq:tldsigma}\eeq

From the proofs of theorems in Section~\ref{sxn:theory}, the change from $\bar{\sigma}$ to $\tilde{\sigma}$ will affect \thref{theorem3} and \thref{theorem5}.
As far as \thref{theorem3} is concerned, if $\tilde{\sigma}_i$ weekly depends on $E_i$, or $\tilde{\sigma}_i$ is a smooth function of $E_i$, the extremum of $g(E)$ determined by condition $dg(E)/dE=0$ remains almost the same as before.
As far as \thref{theorem5} is concerned, if $\sigmast_p$ is replaced by $\tilde{\sigma}^{\ast}_p$, the corresponding luminosity allocation can be obtained.

\subsection{Correlation issue}\label{sxtcorrelation}
Correlation due to systematic uncertainty in scan experiment is always an annoying problem. A so-called scale factor method was used to deal with correlated data~\cite{D'Agostini:1993uj}. The application of such a method in scan experiment is explored in details in Refs.~\cite{moxh2003a,moxh2003b,moxh2007}. The general idea is to introduce a factor corresponding to the correlating uncertainty, so that the factor can be treated as an independent measurement variable. Therefore, the method depicted in the preceding section for independent systematic uncertainty can be adopted. However, more special study is needed for such a case.

\subsection{Optimization issue}
In developing theory of second optimization, many fine analytical properties for the objective function have been assumed in order to make the first optimization feasible, as it is stated in Sect.~\ref{sxn:theory}. Especially in the proof of \thref{theorem3}, the auxiliary function is required to have only one maximum in scan region. In fact, if the auxiliary function have several same maxima, any one of them is equivalently good for parameter determination. So for this case, one point is enough as well.

However, if we come across the case where the parameters contained in the objective function have multiple solution~\cite{Yuan:2009gd,Mo:2010bw,Yuan2012,Zhu:2011ha}, optimization procedure can only be applied for one set of parameters. For the general case involving all sets of parameters, it is a topic for the further investigation.

In Sect.~\ref{sxn:theory}, Gauss-Newton algorithm is adopted for the first optimization, which is a universally utilized method and has very good properties. Especially, affine invariance of step makes the proof of \thref{theorem2} (the theorem of independence of optimal parameters) feasible and easy. As far as many other algorithms are concerned, theory of second optimization should be considered and studied separately.

\subsection{Sampling method and analytical theory}\label{sxtcmpsampandana}
Last but no least, we will say few words about the sampling method and the analytical theory. It is obviously the former provides the important clue and implication before the latter, and accommodate confirmation after the latter. Moreover, the sampling method can perform study on rather more complex and unknown cases which could not be settled by the present analytical theory.

As far as analytical theory is concerned, its merit is prominent. As long as the conclusion is proved mathematically, the relevant issue can be fixed finally. Just as it have been shown, the analytical theory can provide robust and optimal scheme design for scan experiment. Moreover, even if some more generalized and more complex conclusions could not be proved temporarily, proved conclusions can provide us much clues for further exploration. It is evident that two approaches are complementary and of paramount important in developing theory of second optimization for scan experiment.

\section{Summary}\label{sct:sumary}

In this paper, the sampling technique and the analytic analysis are adopted for multi-parameter optimization fitting involving scan data. The conclusions drawn from two approaches are consistent with each other just as expected, that is

\begin{enumerate}
\item For $n$ parameters scan experiment, $n$ energy points are necessary and sufficient for optimal determination of these $n$ parameters;
\item Each optimal position can be acquired by single parameter scan (sampling method), or resort to the analysis of auxiliary function (analytic theory);
\item The luminosity allocation among points can be determined analytically, which is relevant to relative importance between parameters, cross section and its derivative to parameter at optimal point.
\end{enumerate}

Theory of second optimization for scan experiment established in this paper can provide the state of art scheme for scan experiments that aim at accurate measurements of interesting parameters.

\section*{Acknowledgement}
This work is supported in part by National Natural Science Foundation of China (NSFC) under contracts Nos.: 11375206, 10775142, 10825524, 11125525, 11235011; the Ministry of Science and Technology
of China under Contract Nos.: 2015CB856700, 2015CB856706, and the CAS Center for Excellence in Particle Physics (CCEPP).

\setcounter{equation}{0}
\setcounter{section}{0}
\setcounter{theorem}{0}
\renewcommand{\theequation}{A.\arabic{equation}}

\section*{Appendix A}
As to the mathematical details involved in this paper, it can be referred to Refs.~\cite{bk:Qiu2013}, ~\cite{bk:Golan2007}, ~\cite{bk:Wolf1978}, ~\cite{bk:Yuan1997}, and \cite{bk:Boyd2004}. Complied here are some materials from various areas of mathematics that are used or supposed to be satisfied for the proofs in Sect.~\ref{sxn:theory}.

\begin{proposition} If an $n\times n$ matrix $E$ satisfies $\|E\|<1$, then ${\displaystyle \sum\limits^{\infty}_{k=0} E^k = (I-E)^{-1} }$.
\label{apdxprpsn1}\end{proposition}

\begin{theorem} If $A \in R^{n\times n}$ is a symmetric matrix, there exists an orthogonal matrix $U$ such that $\Lambda= U^T A U $ is a diagonal matrix.
\label{apdxprpsn2}\end{theorem}

\begin{proposition} The orthogonal transformation keeps the trace of matrix invariant.
\label{apdxprpsn3}\end{proposition}

\begin{proposition} The inverse of unit upper (lower) triangle matrix is the unit upper (lower) triangle matrix as well.
\label{apdxprpsn4}\end{proposition}

\begin{proposition} An $n\times n$ matrix $E$ is invertible (also nonsingular or non-degenerate) if and only if its determinant is not equal to zero.
\label{apdxprpsn4a}\end{proposition}

\begin{theorem}[Cramer theorem]  If $A \in R^{n\times n}$ is an $n\times n$ nonsingular matrix and if $Y= (y_{1}, y_{2}, \cdots, y_{n})^T \in R^{n}$, then the system of linear equations $AX=Y$ has the unique solution or $X= (x_{1}, x_{2}, \cdots, x_{n})^T \in R^{n}$ in which, for each $1 \leq j \leq n$, we have $x_i=|A|^{-1} |A_{(i)}|$, where $A_{(i)}$ is the matrix formed from $A$ by replacing the $i$-th column of $A$ by $Y$.
\label{apdxprpsn4b}\end{theorem}

\begin{theorem}[Cholesky decomposition theorem]  If $A \in R^{n\times n}$ is a symmetric positive definite matrix, there exists a real lower triangle matrix $L$ such that $A=LL^T$ or $A =\hat{L} D \hat{L}^T$, where $D$ is the diagonal matrix, and $\hat{L}$ is the unit lower triangle matrix.
\label{apdxprpsn5}\end{theorem}

%\begin{theorem} If $S \subset R^{n}$ is a nonempty open convex set, the function $f$ is defined on %$S$ with the second derivative, then the necessary and sufficient condition for $f$ to be a convex %function is that the Hessian matrix of $f$ is positive semi-definite on every point of $S$.
%\label{apdxprpsn6}\end{theorem}

\begin{theorem} Let $f:R^{n}\to R^{1}$ has continuous second partial derivatives in an open convex set $S \subseteq R^{n}$. Then
\begin{enumerate}
\item $f$ is convex in $S$ if and only if the Hesse matrix $G$ of $f$ is positive semi-definite in $S$;
\item $f$ is strictly convex in $S$ if $G$ is positive definite in $S$, but the converse is not in general true.
\end{enumerate}
\label{apdxprpsn6a}\end{theorem}

\begin{theorem}If
\begin{enumerate}
\item $f:R^{n}\to R^{1}$ is strictly convex in the convex set $S$;
\item $f$ has continuous first partial derivatives in $S$;
\item $\thetast$ is a critical point of $f$ in $S$,
\end{enumerate}
then $\thetast$ is strong global minimizer of $f$ over $S$.
\label{apdxprpsn6b}\end{theorem}

\begin{proposition}[Lipschitz continuity] Given two metric spaces ($X, d_X$) and ($Y, d_Y$), where $d_X$ denotes the metric on the set $X$ and $d_Y$ is the metric on set $Y$ (for example, the metric $d_X(x_1, x_2) = \|x_1-x_2\|$), a function $f: X \to Y$ is called Lipschitz continuous if there exists a real constant $\gamma \geq 0 $ such that, for all $x_1$ and $x_2$ in $X$,
$d_Y(f(x_1), f(x_2)) \le \gamma  d_X(x_1, x_2)$. Any such $\gamma $ is referred to as a Lipschitz constant for the function $f$.
\label{apdxprpsn7}\end{proposition}

\begin{theorem}[Intermediate value theorem] Consider an interval $I = [a, b] \subset R$ and a continuous function $f:I\to R$. Then if $u$ is a number between $f(a)$ and $f(b)$, and
$f(a) \leq u \leq f(b)$ or $f(a) \geq u \geq f(b)$, then there is a $c \in [a, b]$
such that $f(c) = u$.
\label{apdxprpsn7a}\end{theorem}

\begin{theorem}[Sufficient condition for existence of extremum] Suppose that $x^0 = (x^0_1, x^0_2, \cdots, x^0_n)$ is the stable point of function $y=f(x) = f(x_1, x_2, \cdots, x_n)$, moreover, in the neighborhood of the stable point $x^0$, function $f(x)$ has definition, continuous, and has the continuous first and second partial derivatives. Introduce a symbol
$$y^{0}_{x^{p_1}_1, x^{p_2}_2, \cdots, x^{p_n}_n} \equiv
\left( \frac{\partial^k y}{\partial x^{p_1}_1, \partial x^{p_2}_2, \cdots, \partial x^{p_n}_n} \right)_{x^0}, ~~k=\sum\limits^{n}_{i=1} p_{i}, $$
the superscript $0$ indicates that the partial derivatives are calculated at point $x^0$. Define the determinant $D_i$ as
\[
D_i=\left|
\begin{array}{cccc}
y^0_{x_1^2}  &y^0_{x_1x_2} &\cdots & y^0_{x_1x_i} \\
y^0_{x_2x_1} &y^0_{x_2^2}  &\cdots & y^0_{x_2x_i} \\
\vdots&\vdots&\ddots&\vdots \\
y^0_{x_ix_1} &y^0_{x_ix_2} &\cdots & y^0_{x_i^2}
\end{array}
\right|,
\]
For $n$ variables, the $n$ determinants $D_1, D_2, \cdots, D_n$ are calculated in turn, then
\begin{enumerate}
\item The sufficient condition for the stable point $x^0$ to be the minimizer is that all determinants are positive, that is
    $$D_i> 0, i=1,2,\cdots, n;$$
\item The sufficient condition for the stable point $x^0$ to be the maximizer is that all even determinants are positive and all odd determinants are negative, that is
\[\begin{array}{ll}
D_i > 0, & i=1,3,5, \cdots , \\
D_i < 0, & i=2,4,6, \cdots . \\
\end{array} \]
\end{enumerate}
If above two conditions are not satisfied, then the stable point may be not the extremum point. If all $D_i$ are zero, the higher order of derivative has to be considered.

\label{apdxprpsn8}\end{theorem}

%\end{appendix}

\end{document}